%% file: main.tex
\pdfoutput=1

\documentclass{ceurart}
\usepackage{bookmark}
\usepackage{hyperref}
\input{settings.tex}
\input{macros}

\makeatletter
\def\desclabel#1#2{\begingroup
\def\@currentlabel{#1}%
#1\label{#2}\endgroup
}
\makeatother

\begin{document}

\input{paper-info.tex}

\begin{abstract}
The literature on ontology-mediated query answering (OMQA) has been shaped by two key results: first-order rewritability for DL-Lite, and PTime-hardness of data complexity for essentially every description logic beyond~it. This has effectively positioned DL-Lite as the only practical choice for query rewriting, restricting OMQA solutions to first-order queries and ontologies that can be rewritten into them. This AC$_0$ vs.\,PTime dichotomy is especially limiting if we consider that OMQA targets graph-structured data, and that standard graph query languages (including the recent ISO standards GQL and SQL/PGQ) are typically NL-complete.  
Towards identifying a rich Horn DL that can be rewritten into graph query languages and that can still express many $\mathcal{ELI}$ and DL-Lite ontologies, 
we 
introduce a stratification mechanism for $\mathcal{ELI}$ that controls the interaction between conjunction and recursion. In this way, we obtain~$\ourEL$, a~description logic that strictly extends the core DL-Lite, supports reachability axioms and restricted conjunction, and allows for reasoning in~NL. We establish the NL upper bound via a rewriting into nested two-way regular path queries, a fragment of~GQL,  providing initial evidence that our ontology language is a promising candidate for extending OMQA to graph query languages.
\end{abstract}


\maketitle

\input{sections/01-introduction}

\input{sections/02-preliminaries}

\input{sections/03-ourlogic}
\input{sections/04-soundness}
\input{sections/05-completeness}
\input{sections/06-combined-complexity}
\input{sections/07-conclusions}

\begin{acknowledgments}
This research was funded in whole by the Austrian Science Fund (FWF) 10.55776/PIN8884924.
\end{acknowledgments}

\section*{Declaration on Generative AI}
  During the preparation of this work, the author(s) 
  used ChatGPT and Claude in order to: grammar and spelling check, polish the text. 
  After using these tool(s)/service(s), the author(s) reviewed and edited 
  the content as needed and take(s) full responsibility for the publication’s content.

\bibliography{references}

\appendix
\input{sections/appendix-soundness}
\input{sections/appendix-completeness}
\input{sections/appendix-pspace}

\end{document}

%% file: settings.tex
\usepackage{listings}
\usepackage{graphbox}
\DeclareRobustCommand{\DLLogo}{%
  \begingroup\normalfont
  \kern-1.75pt\includegraphics[align=c,height=1.25\baselineskip]{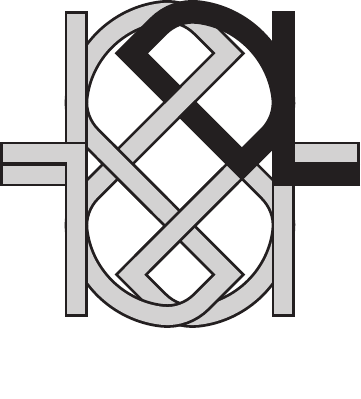}\kern-1.5pt%
  \endgroup
}

\usepackage{url}
\usepackage[utf8]{inputenc}
\usepackage[T1]{fontenc}
\usepackage{graphicx}
\usepackage{amsmath}
\usepackage{amssymb}
\usepackage{amsthm,bm,bbm}
\usepackage{booktabs}
\usepackage{algorithm}
\usepackage{algorithmic}
\usepackage[dvipsnames]{xcolor} 
\usepackage{todonotes}
\usepackage{stmaryrd}
\usepackage{xspace}
\usepackage{paralist} 
\usepackage[mathscr]{euscript}
\usepackage{enumitem}
\usepackage{subcaption}
\usepackage{tikz} \usetikzlibrary{automata, arrows.meta, positioning,
  shapes, fadings}
\usepackage{textcomp}
\usepackage{mathtools}

\usepackage{microtype} 
\usepackage{ellipsis} 
\usepackage[abbreviations,british]{foreign} 
\newcommand{\myqed}{\hfill{$\blacktriangleleft$}}

\DeclareMathAlphabet{\mathdutchcal}{U}{dutchcal}{m}{n}
\SetMathAlphabet{\mathdutchcal}{bold}{U}{dutchcal}{b}{n}
\DeclareMathAlphabet{\mathdutchbcal}{U}{dutchcal}{b}{n}

%% file: macros.tex
\newtheorem{df}{Definition}[section]

\newtheorem{example}[df]{Example}
\newtheorem{theorem}[df]{Theorem}

\newtheorem{lemma}[df]{Lemma}

\newtheorem{cl}[df]{Claim}
\newtheorem{definition}[df]{Definition}
\newtheorem{corollary}[df]{Corollary}

\definecolor{ao(english)}{rgb}{0.0, 0.5, 0.0}

\definecolor{cornflowerblue}{rgb}{0.39, 0.58, 0.93}

\definecolor{brickred}{rgb}{0.8, 0.25, 0.33}

\newcommand{\complexityclass}[1]{\textsc{#1}} 
\newcommand{\NL}{\complexityclass{NL}} 
\newcommand{\PSpace}{\complexityclass{PSpace}} 
\newcommand{\ACzero}{\complexityclass{AC}_0} %

\newcommand{\PTime}{\complexityclass{PTime}} 

\newcommand{\deff}{\coloneqq} 
\newcommand{\N}{\mathbb{N}} 
\newcommand{\rstr}[1]{\hspace{-.2ex}{\downharpoonright}_{#1}} 

\newcommand{\someL}{\mathcal{L}}
\newcommand{\EL}{\mathcal{EL}}

\newcommand{\ELI}{\mathcal{ELI}}
\newcommand{\ELIHbot}{\mathcal{ELIH}^{\bot}}
\newcommand{\ELIbot}{\mathcal{ELI}^{\bot}}
\newcommand{\ourEL}{\mathcal{ELI}^{\bot}_{\preceq}}
\newcommand{\DLLite}{\mathcal{DL}_{\mathrm{lite}}}

\newcommand{\NTwoRPQ}{\mathrm{N2RPQ}}
\newcommand{\IQ}{\mathrm{IQ}}

\newcommand{\lang}[1]{\mathbf{#1}}  
\newcommand{\Ilang}{\lang{N_I}}     
\newcommand{\Rlang}{\lang{N_R}}     
\newcommand{\Clang}{\lang{N_C}}     

\newcommand{\elihbotconc}{\lang{C}_{\ELIHbot}} 

\newcommand{\elibotconc}{\lang{C}_{\ELIbot}} 

\newcommand{\cA}{\mathrm{A}}       
\newcommand{\cB}{\mathrm{B}}       
\newcommand{\cC}{\mathrm{C}}       
\newcommand{\cD}{\mathrm{D}}       
\newcommand{\cE}{\mathrm{E}}       
\newcommand{\cL}{\mathrm{L}}       
\newcommand{\cX}{\mathrm{X}}       

\newcommand{\rr}{\mathit{r}}           
\newcommand{\rs}{\mathit{s}}           

\newcommand{\ia}{\texttt{a}}       
\newcommand{\ib}{\texttt{b}}       
\newcommand{\ic}{\texttt{c}}       
\newcommand{\id}{\texttt{d}}       

\newcommand{\dd}{\mathrm{d}}                             
\newcommand{\de}{\mathrm{e}}                             

\newcommand{\dland}{\sqcap}

\newcommand{\topconcept}{\top}
\newcommand{\botconcept}{\bot}
\newcommand{\dlsubseteq}{\sqsubseteq}

\makeatletter
\providecommand{\bigdland}{%
  \mathop{%
    \mathpalette\@updown\bigsqcup
  }%
}
\newcommand*{\@updown}[2]{%
  \rotatebox[origin=c]{180}{$\m@th#1#2$}%
}
\makeatother

\newcommand{\kb}[1]{\mathcal{#1}}   
\newcommand{\kbK}{\kb{K}}           
\newcommand{\abox}[1]{\mathcal{#1}} 
\newcommand{\aboxA}{\abox{A}}       
\newcommand{\tbox}[1]{\mathcal{#1}} 
\newcommand{\tboxT}{\tbox{T}}       
\newcommand{\ind}[1]{\mathsf{ind}{(#1)}} 
\newcommand{\indA}{\ind{\aboxA}}    
\newcommand{\indK}{\ind{\kbK}}    
\newcommand{\con}[1]{\mathsf{con}{(#1)}} 
\newcommand{\conK}{\con{\kbK}}    
\newcommand{\rol}[1]{\mathsf{rol}{(#1)}} 
\newcommand{\rolK}{\rol{\kbK}}    

\newcommand{\interI}{\mathcal{I}}         
\newcommand{\DeltaI}{\Delta^{\interI}}              
\newcommand{\cdotI}{\cdot^{\interI}}        

\newcommand{\aA}{\mathdutchbcal{A}} 
\newcommand{\aB}{\mathdutchbcal{B}} 

%% file: paper-info.tex
\copyrightyear{2026}
\copyrightclause{Copyright for this paper by its authors.
  Use permitted under Creative Commons License Attribution 4.0
  International (CC BY 4.0).}

\conference{\DLLogo{} DL 2026: 39th International Workshop on Description 
  Logics, July 17--19, 2026, Lisbon, Portugal}

\title{A Horn extension of DL-Lite 
with NL data complexity 
}

\author[1]{Janos Arpasi}[%
email=janos.arpasi@tuwien.ac.at
]

\author[1,2]{Bartosz Jan Bednarczyk}[%
orcid=0000-0002-8267-7554,
email=bartek@cs.uni.wroc.pl,
url=https://bartoszjanbednarczyk.github.io,
]

\author[1]{Magdalena Ortiz}[%
orcid=0000-0002-2344-9658,
email=magdalena.ortiz@tuwien.ac.at
]

\address[1]{Institute of Logic and Computation, TU Wien}
\address[2]{Computer Science Department, University of Wrocław}

%% file: sections/01-introduction.tex
\section{Introduction}

One of the seminal results in the field of Ontology-Mediated Query Answering 
(OMQA)
is the $\ACzero$ data complexity of $\DLLite$, which facilitates pure query rewritings into First-Order Logic queries \cite{calvanese_tractable_2007}. Specifically, given an ontology-mediated query $(q, \tboxT)$, where $q$ is a conjunctive query and $\tboxT$ is a $\DLLite$-TBox, one can derive an FO-query $q_{\tboxT}$ (specifically, a union of conjunctive queries) such that the answers to $q_{\tboxT}$ and $(q, \tboxT)$ coincide over any ABox $\mathcal{A}$. This result effectively reduces ontology-mediated query evaluation to the evaluation of standard database queries, bypassing the need for specialised reasoning engines at runtime.
Equally influential is the observation that $\DLLite$ is essentially the only major family of description logics (DLs) admitting such FO-rewritings \cite{calvanese_data_2013}. Even relatively inexpressive logics, such as $\mathcal{EL}$, are $\PTime$-complete in data complexity \cite{calvanese_data_2013}. Consequently, they cannot be rewritten into FO, nor into any query language with subpolynomial evaluation complexity. This delineates a long-standing barrier that has shaped much of the research landscape in OMQA for the past two decades.

However, this barrier may have been accepted too readily. There is a range of graph query languages that can express fundamental graph features such as reachability and path navigation, surpassing the expressive power of FO while maintaining data complexity in $\NL$ \cite{DBLP:conf/pods/LibkinMMPV25}. 
These languages have matured significantly over the past decade, culminating in the standardisation of GQL and its SQL-inspired counterpart, SQL/PGQ, in 2023 \cite{DBLP:conf/icdt/FrancisGGLMMMPR23}. 
These languages are obvious candidates for OMQA, making it a true graph data access paradigm and opening up the possibility of finally moving beyond DL-Lite and FO-rewritability.
The foundations of ontology-mediated queries based on graph query language---specifically, conjunctive regular path queries (CRPQs) and their variants--- were laid over a decade ago, and the landscape of their computational complexity was almost fully charted  \cite{bienvenu_regular_2015,BienvenuCOS14,DBLP:conf/ijcai/CalvaneseEO09,DBLP:conf/ijcai/OrtizRS11}. Yet these results never led to the attempted adoption of ontology-mediated graph queries in practical systems. Furthermore, none of the mentioned works addressed the fundamental tension between the limited expressivity of $\DLLite$ and the loss of rewritability in more expressive settings.

Recent advances in graph query languages offer new momentum to overcome this impasse. Finally, state-of-the-art graph query engines (like Neo4j and its query language Cypher \cite{DBLP:conf/icdt/GheerbrantLPR25}) are moving towards full support of the navigational core of GQL and SQL/PGQ,\footnote{\url{https://neo4j.com/docs/cypher-manual/current/appendix/gql-conformance/}} which supports all the extensions of CRPQs considered so far in the OMQA literature, including \emph{nested 2-way conjunctive regular path queries (N2RPQs)}, their most expressive variant \cite{DBLP:conf/pods/FrancisGGLMMMPR23}. 
We now see a promising path towards the ultimate goal of \emph{designing useful, expressive yet practicable ontology-mediated graph query languages} that allow us to leverage existing ontologies in the OWL profiles and to pair them with the flexible navigational features of graph query languages. 

In this paper, we address a specific aspect of this goal that has been shamefully disregarded in the literature, namely the identification of Horn DLs that strictly extend $\DLLite$ and admit rewritings into graph query languages. Largely discouraged by the early negative results---like the mentioned $\PTime$-hardness of 
plain $\EL$, and the same lower bound for any DL supporting qualified existential restrictions in the left-hand side of axioms and inverse roles \cite{calvanese_data_2013}---hardly any research effort has been invested into cutting out meaningful description logics that admit rewritings into CRPQs and their extensions. 
  To our knowledge, there are only two partial exceptions. 
  \emph{Harmless linear $\mathcal{ELHI}$} \cite{dimartino_efficient_2025} disallows the use of concept conjunction in the left-hand side of axioms, essentially restricting $\mathcal{ELHI}$ to axioms of the forms 
  $\cA \dlsubseteq \cB$, $\cA \dlsubseteq \exists{\rs}.\cB$ and  $\exists{\rs}.\cA \dlsubseteq \cB$, with $\rs$ a possibly inverse role, and imposing an additional \emph{harmlessness} constraint that restricts the interaction of inverses and qualified existential restrictions to keep the complexity in $\NL$.  
  The resulting description logic admits rewriting of atomic queries into a simple class of CRPQs, but the lack of conjunction makes it very limited as a Horn DL.  
Motivated by the central role of conjunction in many real-world ontologies, quasi-linear $\mathcal{ELHI}$ takes a slightly different approach, enabling restricted use of conjunction at the expense of additional restrictions on the use of inverse roles\cite{curry_towards_2025}. 
The resulting ontology language appears to be useful in practice, but it is not a proper extension of $\DLLite$. 
We aim to push the boundary of description logics that admit reasoning in $\NL$ and rewritings into graph query languages, extending $\DLLite$ without compromising too much of the additional expressiveness of $\mathcal{ELI}$. In particular, instead of disallowing the use of the conjunction, we leverage the expressiveness of N2RPQs to support it in a controlled way. 

In this paper, we propose $\ourEL$ as the first $\NL$-data-complete ontology language that fully subsumes $\DLLite$ and that supports conjunction and reachability axioms of the form $\exists{\rr}.\cA \dlsubseteq \cA$. To maintain $\NL$ data complexity, $\ourEL$ utilises a stratification of concept names that guarantees a bound on the alternation induced by the interaction of reachability, inverses, and conjunction, preventing the complexity from escalating to $\PTime$. We demonstrate $\NL$-membership for instance checking via a rewriting into N2RPQs. 
Finally, we show that while the data complexity is low, the combined complexity of $\ourEL$ is $\PSpace$-complete; this is consistent in spirit with the restriction of $\complexityclass{ExpTime}$ 
logics to linear recursion \cite{dantsin_gottlob_2001}, and witnesses the additional expressive power of our fragment in comparison to  $\DLLite$, $\EL$, and harmless linear $\mathcal{ELHI}$,  all of which allow for reasoning in polynomial time.

%% file: sections/02-preliminaries.tex

\section{Preliminaries}\label{sec:preliminaries}

We start by recalling the basics on description logics (DLs)~\cite{dlbook} 
and query answering~\cite{OrtizS12}.

\paragraph*{DLs.}\label{subsec:prelim-dls}
  We fix countably infinite pairwise disjoint sets of 
  \emph{individual names} \( \Ilang \), \emph{concept names} \( \Clang \), 
  and \emph{role names}~\( \Rlang \) and introduce the description 
  logic \( \ELIbot \). Starting from \( \Clang \) and \( \Rlang \), 
  the set~\( \elibotconc \) of $\ELIbot$-\emph{concepts} is built using 
  the following concept constructors: 
  \emph{conjunction} \((\cC \dland \cD) \), 
  \emph{existential restriction} (\(\exists{\rr}.\cC \)) 
  and \emph{bottom} and \emph{top concepts} (\( \botconcept, \topconcept \)), 
  with the following grammar:

  \begin{equation*} \label{eq:elihbot-grammar}
  \cC, \cD \; \Coloneqq 
  \; \botconcept \; \mid 
  \; \topconcept \; \mid \; 
  \cA \; \mid \; \cC \dland \cD \; 
  \mid \; \exists{\rr}.\cC \; \mid \; \exists{\bar{\rr}}.\cC,
  \end{equation*}
  where $\cC$, $\cD$ are $\elihbotconc$-concepts, $\cA \in \Clang$ is a concept name, 
  $\rr \in \Rlang$ is a role name, and $\bar{\rr}$ denotes its inverse.
  We will always assume that the inverse of $\bar{\rr}$ is just $\rr$, and use (possibly decorated) letters $\rs$ to denote possibly inverted roles. 

\paragraph*{KBs.}\label{subsec:prelim-KBs}
  \emph{Assertions} take the form \( \cC(\ia) \), \(\rs(\ia,\ib) \)  for \(\ia,\ib \in \Ilang \), \(\cC \in \Clang \cup \{ \top, \bot \} \) and (possibly inverted) role \(\rs \).
  A~\emph{general concept inclusion} (GCI) has the form \(\cC \dlsubseteq \cD \) for concepts \(\cC, \cD \). 
  We employ $\cC \equiv \cD$ as a shorthand for the two GCIs \(\cC \dlsubseteq \cD \) and \(\cD \dlsubseteq \cC \).
  A~\emph{knowledge base} (KB) \(\kbK \deff (\aboxA, \tboxT) \) is composed of a finite non-empty set~\(\aboxA \) (\emph{ABox}) of assertions and a finite non-empty set \(\tboxT \) (\emph{TBox}) of GCIs. 
  The elements of \(\aboxA \cup \tboxT \) are called \emph{axioms}. 
  Let \( \indK \), \( \conK \), \( \rolK \)  denote the sets of all individual, concept, and role names~from~\( \kbK \), respectively (including $\top, \bot$ in $\conK$ and role inverses in $\rolK$). We use the analogue notation for TBoxes and ABoxes as well.
  A KB \(\kbK \) is in \emph{normal form} if all its GCIs conform to the following pattern:
  \[
    \cA \dlsubseteq \cB, \quad 
    \cA \dland \cB \dlsubseteq \cC, \quad 
    \cA \dlsubseteq \exists{\rs}.\cB, \quad 
    \exists{\rs}.\cA \dlsubseteq \cB, 
  \]
  where $\cA, \cB, \cC$ are either concept names, $\top$, or $\bot$, while $\rs$ is a possibly inverted role.
  Without loss of generality, we also require that $\top$ and $\bot$ do not appear in GCIs of the form $\cA \dland \cB \dlsubseteq \cC$, and that trivially unsatisfiable formulae like $\exists{\rs}.\bot$ are always replaced with $\bot$.

  \begin{table}[!htb]
    \begin{minipage}{.65\linewidth}
      \caption{Concepts and roles in \( \ELIbot \).}\label{tab:ALCcap}
      \centering
          \begin{tabular}{@{}l@{\ \ \ }c@{\ \ \ }l@{}}
              \hline\\[-2ex]
              Name & Syntax & Semantics \\ \hline \\[-2ex]
              bottom concept & \( \botconcept \) & \( \emptyset  \) \\
              top concept & \( \topconcept \) & \( \DeltaI  \) \\
              concept name & \(\cA \) & \(\cA^\interI \subseteq \DeltaI  \) \\ 
              role name & \(\rr \) & \(\rr^\interI \subseteq \DeltaI {\times} \DeltaI \) \\ 
              inverse role & \( \bar{\rr} \) & \( \{ (\de, \dd) \mid (\dd, \de) \in \rr^{\interI} \} \) \\ 
              conc.\ negation & \(\neg\cC \)& \(\DeltaI \setminus \cC^{\interI} \) \\  
              conc.\ intersection & \(\cC \dland \cD \)& \(\cC^{\interI}\cap \cD^{\interI} \) \\  
              exist.\ restriction & \(\exists{\rs}.\cC \) & 
              \(\{ \dd \mid \exists{\de}.(\dd,\de)\in \rs^{\interI} \land \de \in \cC^{\interI} \} \)
              \\\hline
          \end{tabular}
    \end{minipage}%
    \begin{minipage}{.35\linewidth}
    \caption{Axioms in \( \ELIbot \).}\label{tab:axm}
      \centering
          \begin{tabular}{ l l }
              \hline\\[-2ex]
              Axiom \(\alpha \) & \(\interI \models\alpha \), if \\ \hline \\[-2ex] 
              \(\cC \dlsubseteq \cD \) & \(\cC^{\interI} \subseteq \cD^{\interI}  \)\hspace{5ex} \mbox{TBox}~\(\tboxT \) \\\hline \\[-2ex]
              \(\cC(\ia) \) & \(\ia^{\interI} \in \cC^{\interI}  \)\hfill \mbox{ABox}\(~\aboxA \)\hspace{-5ex} \\
              \(\rr(\ia,\ib) \) & \((\ia^{\interI}, \ib^{\interI} )\in \rr^{\interI}  \)
              \\\hline
          \end{tabular}
    \end{minipage} 
  \end{table}

  The semantics of \(\ELIbot \) is defined via \emph{interpretations} 
  \(\interI \deff (\DeltaI, \cdotI) \) composed of a non-empty set~\(\DeltaI \) 
  called the \emph{domain of \(\interI \)} and an \emph{interpretation function} 
  \(\cdotI \) mapping individual names to elements of~\(\DeltaI \), concept 
  names to subsets of \(\DeltaI \), and role names to subsets of \(\DeltaI \times \DeltaI \). 
  This mapping is extended to complex concepts (see Table~\ref{tab:ALCcap}) and finally 
  used to define \emph{satisfaction} of assertions and GCIs (see Table~\ref{tab:axm}). 
  \emph{Structures} are interpretations with a partial assignment of individual names.
  We say that an interpretation \(\interI \) \emph{satisfies} a KB  
  \(\kbK \deff (\aboxA,\tboxT) \) (or \(\interI \) is a \emph{model} of 
  \(\kbK \), written: \(\interI \models \kbK \)) if it satisfies all 
  axioms of~\(\aboxA\cup\tboxT \). 
  A KB is \emph{consistent} (or \emph{satisfiable}) if it has a model 
  and \emph{inconsistent} (or \emph{unsatisfiable})~otherwise. 

\paragraph*{Automata and paths.}\label{subsec:prelim-automata}

By a \emph{nondeterministic automaton} (NFA) $\aA$ we understand, as usual, a tuple $(Q, \Sigma, \delta, q_0, F)$ where \(Q\) is a finite set of states, \(\Sigma\) is a finite alphabet, \(\delta \subseteq Q \times \Sigma \times Q\) is the transition relation, \(q_0 \in Q\) is the initial state and \(F \subseteq Q\) is the set of accepting states. 
All NFA considered in this paper are over finite subsets of $\{ \rr, \bar{\rr}, \cC? \mid \rr \in \Rlang, \cC \in \Clang \}$, where the expressions of the form $\cC?$ are called \emph{tests}.
For a natural number $n$ we also define \emph{$n$-nested NFA} (nNFA) by induction: a 0NFA is just an NFA, and an $(n{+}1)$NFA is an NFA whose alphabet may contain letters $\aA?$ for $\aA$ being an nNFA.

By a \emph{run} of an NFA $\aA = (Q, \Sigma, \delta, q_0, F)$ over an interpretation $\interI$, we understand a finite sequence of tuples $(\dd_0, q_0, w_1, q_1, \dd_1), \ldots, (\dd_{\ell-1}, q_{\ell-1}, w_\ell, q_\ell, \dd_\ell)$ in $\DeltaI \times \delta \times \DeltaI$ such that for all $i < \ell$ we have:
\begin{enumerate}[label=(\roman*)]
  \item if $w_{i+1} = \rr$, then $(\dd_i, \dd_{i+1}) \in \rr^{\interI}$,
  \item if $w_{i+1} = \bar{\rr}$, then $(\dd_{i+1}, \dd_i) \in \rr^{\interI}$, and
  \item if $w_{i+1} = \cC?$ for a concept name $\cC$, then $\dd_i = \dd_{i+1}$ and $\dd_i \in \cC^{\interI}$.
\end{enumerate}
A run is \emph{accepting} if $q_\ell \in F$. We say that a run \emph{starts from} $\dd$ if $\dd_0 = \dd$.
We define \emph{runs} of an $n$NFA for all $n > 0$ analogously, adding the following extra condition:
\begin{enumerate}[label=(\roman*)]
  \setcounter{enumi}{3}
  \item if $w_{i+1} = \aB?$ for an $(n{-}1)$NFA $\aB$, then $\dd_i = \dd_{i+1}$ and there is an accepting run of $\aB$ starting from $\dd_i$.
\end{enumerate}


\paragraph*{Queries.}\label{subsec:prelim-query}
An \emph{instance query} ($\IQ$) $q$ is an expression of the form $\cC(\ia)$, where $\ia$ is an individual name and $\cC$ is a concept name. An interpretation $\interI$ satisfies such a query (written $\interI \models q$) if $\ia^{\interI} \in \cC^{\interI}$.
A~\emph{nested two-way regular path query} ($\NTwoRPQ$) $q$ is an expression of the form $\exists{x}\aA(\ia,x)$ for an individual name~$\ia$ and an nNFA $\aA$. An interpretation $\interI$ satisfies such a query (written $\interI \models q$) if there is an accepting run of $\aA$ in~$\interI$ starting from $\ia^{\interI}$.
In the \emph{query entailment problem}, we are given a KB $\kbK \deff (\aboxA, \tboxT)$ and a query~$q$ and ask if $\kbK \models q$, i.e. whether every model of $\kbK$ satisfies $q$.
Here we are interested in computational complexity of this problem. If all of $\aboxA$, $\tboxT$ and $q$ are part of the input, we talk about \emph{combined complexity}. If only $\aboxA$ is part of the input, while $\tboxT$ and $q$ are fixed beforehand, we talk about \emph{data complexity}.
Given a query language $\someL$ (say, first-order logic or $\NTwoRPQ$) say that a TBox~$\tboxT$ is $\someL$-\emph{rewritable} if for every instance query $q$ there is a query $q'$ in $\someL$ such that for all ABoxes $\aboxA$ and all~$\ia \in \indA$, we have $(\aboxA, \tboxT) \models q(\ia)$ if and only if $\aboxA \models q'(\ia)$. 
Briefly speaking, the query entailment problem over KBs with $\someL$-\emph{rewritable} TBoxes boils down to just evaluating $\someL$-queries over ABoxes.

%% file: sections/03-ourlogic.tex

\section{Our Logic and the Rewriting}\label{sec:our-logic-and-rewriting}


In this section, we introduce $\ourEL$ and present the rewriting of 
instance queries into $\NTwoRPQ$s. 

\begin{definition}\label{def:stratified-kb}
  We say that a KB \(\kbK = (\aboxA, \tboxT) \) in normal form is \emph{stratified} 
  whenever there exists a preorder $\preceq$ on $\conK \cup \rolK$ satisfying all the following conditions (where $\prec$ denotes the strict part of $\preceq$, and $\cA, \cB, \cC$ are concept names, and $\cD$ stands for a concept name or $\top$):
  \begin{itemize}
    \item Whenever $(\cA \dlsubseteq \cB) \in \tboxT$, then $\cA \preceq \cB$.
    \item Whenever $(\cA \dland \cB \dlsubseteq \cC) \in \tboxT$, then $\cA \preceq \cC$, $\cB \preceq \cC$, and at least one of $\cA \prec \cC$ or $\cB \prec \cC$ holds.
    \item Whenever $(\cA \dlsubseteq \exists{\rs}.\cD) \in \tboxT$, then $\cA \preceq \cD$ or $\cD = \top$; and $\cA \preceq \rs$.
    \item Whenever $(\exists{\rs}.\cD \dlsubseteq \cB) \in \tboxT$, then either $\cD \prec \cB$, $\cD = \cB$, or $\cD =\top$; and $\rs \preceq \cD$.
    \item We assume that $\rr \preceq \rr^-$ and $\rr^- \preceq \rr$ for all role names $\rr$.
  \end{itemize}
  By $\ourEL$ we denote the fragment of $\ELIHbot$ whose KBs are in normal form and stratified.\myqed
\end{definition}
Note that the above conditions on $\preceq$ do not apply to GCIs with $\bot$ on the right-hand side, and hence disjointness of concepts is not subject to any restriction in $\ourEL$.
 It is easy to see that $\ourEL$ strictly extends $\DLLite$~\cite[Def.~7.9]{dlbook} while being a strict fragment of $\ELIHbot$.
A separating example is the non-first-order-rewritable TBox $\tboxT \deff \{ \exists{\rr}.\cA \dlsubseteq \cA \}$, which belongs to our logic but not to $\DLLite$~\cite[Consequence of Thm.~7.8]{dlbook}.
On the other hand, $\tboxT' \deff \{\exists{\rr}.\cA \dland \exists{\rs}.\cA \dlsubseteq \cA \}$ is a TBox that belongs to $\ELIbot$ but not to $\ourEL$.
The rewritings of the instance query $\cA(x)$ over both $\tboxT$ and $\tboxT'$ are inherently recursive and thus fall outside first-order logic.
A reader familiar with Datalog~\cite[Sec.~7.3]{dlbook} will recognize that the rewriting for $\tboxT$ fits into the fragment of Datalog with linear recursion, whereas the one for $\tboxT'$ does not.
Consequently, the rewriting of $\tboxT'$ into nested automata or well-known graph query languages such as GQL is not possible, while the rewriting of $\tboxT$ can be expressed as a regular path query, \eg as $(\rr^*; \cA?)(\ia,x)$.
Requiring recursion to be linear is a crucial condition for guaranteeing the existence of rewritings of instance queries.
In our approach, this is enforced by the preorder $\preceq$ in Definition~\ref{def:stratified-kb} (in particular by the second condition).

\begin{example}
It is a well-known fact that in expressive DLs such as $\ELI$, one can simulate concept conjunction by means of role inverses.
Indeed, the GCI $\cA \dland \cB \dlsubseteq \cC$ can be replaced by the following set of GCIs
$\tboxT \deff \{ \cA \dlsubseteq \exists{\rr}.\top,\; \exists{\bar{\rr}}.\cB \dlsubseteq \cD,\; \exists{\rr}.\cD \dlsubseteq \cC \}$,
where the role $\rr$ and the concept name $\cD$ are fresh.
We stress that reusing this ``conjunction emulation'' approach for arbitrary concepts fails in our logic due to the presence of the built-in order $\preceq$.
Indeed, it is routine to check that, in the presence of $\preceq$, the following facts hold for $\tboxT$:
$\cA \preceq \rr$, $\rr \preceq \cB$, $\cB \prec \cD$, $\rr \preceq \cC$, $\cD \prec \cC$.
In particular, this implies $\cA \prec \cC$ and $\cB \prec \cC$, and thus $\cA$ and $\cB$ cannot be arbitrary concepts.\myqed
\end{example}

The main technical contribution of this paper is the rewriting of instance queries under $\ourEL$-ontologies into nested two-way regular path queries. We state the result formally below.

\begin{theorem}\label{thm:rewriting}
  Take a $\ourEL$-TBox $\tboxT$ and an instance query $q$.
  We can compute an $\NTwoRPQ$ $q'$ such that for all ABoxes $\aboxA$ 
  and all $\ia \in \indA$,
  we have $(\aboxA, \tboxT) \models q(\ia)$ if and only if $\aboxA \models q'(\ia)$.~\myqed
\end{theorem}

In the remainder of the section, we fix $q$ and $\tboxT$ as in the statement of Theorem~\ref{thm:rewriting}, and present how to construct the desired query $q'$.\footnote{W.l.o.g. we assume that the concept name from $q$ occur in $\tboxT$; otherwise the entailment is trivially false.}
Soundness and completeness of the rewriting are established in the next two sections, respectively.
Our construction of $q'$ is inductive and based on the \emph{height} of the concept/role names $\alpha$ in $\tboxT$, namely the length of the longest chain of concept or role names $\alpha_1 \prec \alpha_2 \prec \ldots \prec \alpha_n$ for which~$\alpha_n = \alpha$.
This notion is lifted to concepts from $\tboxT$ as follows: the height of $\top$ and $\bot$ is $0$, and the height of concepts of the form $\exists{\rs}.\cA$ is the maximum of heights of $\rs$ and $\cA$.
For convenience, let $\tboxT\rstr{n}$ and $\aboxA\rstr{n}$ respectively denote the set of GCIs and assertions from $\tboxT$ and $\aboxA$ that solely concepts of height at most $n$.
Our inductive claim is stated below ($\tboxT\rstr{{-}1}$ is just the empty set).

\begin{center}
\textbf{Inductive assumption for $n \in \N$:} For each concept name $\cA$ in $\tboxT$ of height $n$, we can construct an $n$-nested NFA $\aA_{\cA}$ over $\{ \cB? \mid \cB \in \con{\tboxT\rstr{n}} \} \cup  \rol{\tboxT\rstr{n}} \cup \{ \aA_{\cB}? \mid \cB \in \con{\tboxT\rstr{n{-}1}} \}$ such that for all ABoxes $\aboxA$ and all $\ia \in \indA$, we have $(\aboxA, \tboxT) \models \cA(\ia)$ if and only if $\aboxA \models \exists{x} \aA_{\cA}(\ia, x)$.
\end{center}

Less formally, our goal is to construct an automaton that witnesses the satisfaction of $\cA$ at a given node in an ABox, using the GCIs in $\tboxT$ as inference rules.
Its alphabet consists of (possibly inverted) role names, tests over concept names of height at most $n$ (enabling the automaton to read labels from the ABox), and inductively constructed automaton-tests detecting membership in concepts of height strictly less than $n$ (enabling it to delegate the verification of certain concepts to ``simpler'' automata).
The states are pairs $(\mathit{premise}, \mathit{goal})$, where the first component is a set of concept names from $\tboxT$ and the second one is a single one of those.
Along a run, the automaton tracks the concepts from $\tboxT$ already known to hold at the current node (the \emph{premise}), as well as those that remain to be verified there (the \emph{goal}).
Transitions are either derived from the GCIs in $\tboxT$, mimicking a (reversed) entailment proof via successive rule applications, or read off from the ABox, recording which concept names hold at the current node.
For instance, whenever a GCI $\cA \dlsubseteq \cB$ is present in $\tboxT$ and $\cB$ belongs to the goal, the automaton may replace $\cB$ in its goal with $\cA$.
Finally, a state is accepting if and only if its goal component is empty, \ie all required concepts have been verified; in other words, a run is accepting if the automaton successfully discharges all its goals.

\paragraph*{Base case.}
We handle the base case first, as it already captures the key challenges of our technique.
For the reader's convenience, we interweave formal definitions with informal descriptions.

\begin{definition}\label{def:automaton-base-case}
  Fix a concept name $\cA \in \con{\tboxT\rstr{0}}$.
  We define an NFA $\aA_{\cA} \deff (Q, \Sigma, \delta, q_0, F)$ as follows.
  \begin{itemize}[itemsep=0em]

  \item We set $\Sigma \deff \{ \cB? \mid \cB \in \con{\tboxT\rstr{0}} \} \cup \rol{\tboxT\rstr{0}}$, as already stated in the inductive claim. Note that in the base case no automaton-tests are used.

  \item The state set $Q$ consists of all pairs $(\mathit{premise}, \mathit{goal})$ for $\mathit{premise} \subseteq \con{\tboxT\rstr{0}}$ and $\mathit{goal} \in \con{\tboxT\rstr{0}}$.

  \item The initial state is $q_0 \deff (\{ \top \}, \cA)$.
    The intuition is straightforward: we wish to verify $\cA$ at the current node, but we currently only know that $\top$ holds there.

  \item The set of accepting states $F$ consists of all pairs in $Q$ where $\mathit{goal} \in \mathit{premise}$ or $\bot \in \mathit{premise}$.

  \item The transition relation $\delta$ is defined as the minimal subset of $Q \times \Sigma \times Q$
    such that, for every state $q \deff (\mathit{premise}, \mathit{goal})$, all the following conditions hold.

  \begin{description}[itemsep=0em]

    \item[\desclabel{(weak)}{trans:weak}]
    For any subset $\mathit{premise}'$ of $\mathit{premise}$, there is a transition from $q$ to $(\mathit{premise}' \cup \{\top\}, \mathit{goal})$ labelled~``$\top?$''. 
    This allows the automaton to forget some of the information in the premise, which will be crucial for handling GCIs with conjunctions in the premise.

    \item[\desclabel{(data)}{trans:data}]
      For any concept name $\cB \in \con{\tboxT\rstr{0}}$, there is a transition from $q$ to
      $(\mathit{premise} \cup \{\cB\}, \mathit{goal})$ labelled~``$\cB?$''.
      This allows the automaton to read from the ABox which concept names hold at the
      current node and add them to the premise.

    \item[\desclabel{(sbus)}{trans:sbus}]
      For any GCI $\cC \dlsubseteq \cD$ from $\tboxT\rstr{0}$, if $\mathit{goal} = \cD$,
      there exists a transition from~$q$ to
      $(\mathit{premise},\, \cC)$
      labelled~``$\top?$''.
      This formalizes the following backward reasoning: to prove $\cD$ given the GCI
      $\cC \dlsubseteq \cD$, it suffices to prove $\cC$.

    \item[\desclabel{(succ)}{trans:succ}]
      If
      $\mathit{goal} = \cD$, then for any GCI
      $\exists{\rs}.\cC \dlsubseteq \cD$ in $\tboxT\rstr{0}$, there is a transition from $q$ to
      $(\{ \top \},  \cC )$ labelled~``$\rs$''. 
      This formalizes the following reasoning: given the GCI $\exists{\rr}.\cC \dlsubseteq \cD$,
      the automaton may move to a neighbouring node in the ABox and verify that $\cC$ holds there.
      The premise is reset to $\{ \top \}$ since the automaton does not have any non-trivial information on that neighbour.
      The fact that it suffices to track a single successor at a time---which is essential for the
      feasibility of our construction---follows from the linear recursion enforced by~Definition~\ref{def:stratified-kb}.

    \item[\desclabel{(anon)}{trans:anon}]
      For any concepts $\cB$ and $\cD$ from $\con{\tboxT\rstr{0}}$, if the KB
      $(\{ \cC(\ia) \mid \cC \in \mathit{premise} \cup \{\cB\}\}, \tboxT\rstr{0})$ entails $\cD(\ia)$ and $\mathit{goal} = \cD$,
      there is a transition from $q$ to $(\mathit{premise} , \cB)$
      labelled~``$\top?$''.
      Though this transition may appear redundant, it is crucial for correctness:
      it enables the automaton to infer the existence of anonymous individuals satisfying
      certain concepts, which is essential for handling GCIs with existential restrictions
      in the premise.

  \end{description}

  Note that GCIs of the form $\cA \dland \cB \dlsubseteq \cC$ play no role in the definition
  of $\delta$ here. By the second condition of Definition~\ref{def:stratified-kb}
  and by the definition of a normal form, such GCIs cannot solely involve concepts from $\con{\tboxT\rstr{0}}$.~\myqed
  \end{itemize}
\end{definition}

We illustrate our construction of $\aA_{\cA}$ for a (non-first-order-rewritable) TBox $\tboxT \deff \{ \exists{\rr}.\cA \dlsubseteq \cA \}$ and an instance query $q \deff \cA(\ia)$. The relevant \emph{fragment} of $\aA_{\cA}$ is depicted below. It is equivalent to the regular path query $(\rr^*;\cA?)(\ia, x)$, the expected rewriting for this query and~TBox.

\begin{center}
\begin{tikzpicture}[
    shorten >=1pt,
    node distance=3cm,
    on grid,
    auto,
    every state/.style={draw, circle, inner sep=2pt, minimum size=1.8cm}
]

\node[state, initial]   (q1)                  {\scalebox{0.55}{$\{ \top \}, \{ \cA \}$}};
\node[state, accepting] (q2) [right=of q1]    {\scalebox{0.55}{$\{ \top, \cA \}, \{ \cA \}$}};

\path[->]
    (q1) edge                node {$\cA?$}  (q2)

    (q1) edge [loop, out=225, in=200, looseness=5]    node {$r,\top?$}   (q1)
    (q2) edge [loop, out=225, in=200, looseness=5]    node {$\cA?,\top?$}   (q2);

\end{tikzpicture}
\end{center}

\paragraph*{Inductive step.}
We now proceed with the inductive step. The main difference with the base case is that we can now also use automaton-tests in the transitions to deal with GCIs involving conjunction.

\begin{definition}\label{def:automaton-general-case}
  Fix $n > 0$ and a concept name $\cA \in \con{\tboxT\rstr{n}}$, and assume that $\aA_{\cB}$ has already
  been defined for each $\cB \in \con{\tboxT\rstr{n-1}}$.
  We define a nested NFA $\aA_{\cA} \deff (Q, \Sigma, \delta, q_0, F)$ as follows.
  \begin{itemize}[itemsep=0em]

    \item The components $\Sigma$, $Q$, $q_0$, and $F$ are defined as in the base case, with the
      sole difference that $\Sigma$ now also contains automaton-tests of the form $\aA_{\cB}?$
      for each $\cB \in \con{\tboxT\rstr{n-1}}$, and all occurrences of $\tboxT\rstr{0}$ are replaced with $\tboxT\rstr{n}$.

    \item The transition relation $\delta$ is the minimal subset of $Q \times \Sigma \times Q$
      such that, for every state $q \deff (\mathit{premise}, \mathit{goal})$, all conditions
      from the base case (after replacing $\tboxT\rstr{0}$ with $\tboxT\rstr{n}$) hold, as well as the following three conditions are satisfied.

    \begin{description}    

    \item[\desclabel{(noc)}{trans:noc}]
    For any GCI $\cB \dland \cC \dlsubseteq \cD$ in $\tboxT\rstr{n}$ with $\cB \in \mathit{premise}$ and $\mathit{goal} = \cD$,
    there exists a transition from~$q$ to $(\mathit{premise},  \cC )$
    labelled ``$\top?$''.
    The case of the GCI $\cB \dland \cC \dlsubseteq \cD$ with $\cC \in \mathit{premise}$ and $\cD \in \mathit{goal}$ is handled analogously.
    This formalises the standard backward reasoning step for conjunction.

    \item[\desclabel{(aut)}{trans:aut}]
    For any concept name $\cB \in \con{\tboxT\rstr{n-1}}$ there is a transition 
    from $q$ to $(\mathit{premise} \cup \{ \cB \}, \mathit{goal})$ 
    labelled~``$\aA_{\cB}?$''. 
    As the height of $\cB$ is strictly less than that of $\cA$, the automaton can ``delegate'' the proof of $\cB$ to the nested automaton $\aA_{\cB}$ (well-defined by the inductive hypothesis).\myqed
    \end{description}
  \end{itemize}
\end{definition}

%% file: sections/04-soundness.tex

\section{Soundness}\label{sec:soundness}

In this section we establish the ``only if'' direction of Theorem~\ref{thm:rewriting}, restated for convenience below.

\begin{lemma}\label{lemma:rewriting-soundness}
  Let $\tboxT$ be a stratified $\ourEL$-TBox, 
  $\cA(a)$ be an instance query, and
  $\aA_{\cA}$ be an nNFA from Definition~\ref{def:automaton-general-case}.
  Then for all for all ABoxes $\aboxA$ we have that 
  $\aboxA \models \exists{x} \aA_{\cA}(\ia, x)$ implies 
  $(\aboxA, \tboxT) \models \cA(\ia)$.~\myqed
\end{lemma}
\begin{proof}
Fix a TBox $\tboxT$. 
The proof is by induction on
the height $n$ of the concept name $A$.
We assume that for all $n' < n$ and concept names $\cB \in \con{\tboxT\rstr{n'}}$, the statement of Lemma~\ref{lemma:rewriting-soundness} holds for the instance query $\cB(\ia)$. 

Then we consider 
an instance query $\cA(\ia)$ with $\cA \in \con{\tboxT\rstr{n}}$.
Fix an ABox $\aboxA$, a model $\interI$ of $(\aboxA, \tboxT)$, and let $\rho = \rho_1 \cdots \rho_\ell$ be an accepting run of $\aA_{\cA}$ in $\interI$ starting from $\ia^{\interI}$, which exists by entailment $\aboxA \models \exists{x} \aA_{\cA}(\ia, x)$).
For convenience, let $\rho_i$ unfold to $(\dd_{i-1}, \delta_i, \dd_i)$ with $q_i = (\mathit{premise}_i, \mathit{goal}_i)$ and $\delta_i = (q_{i-1}, w_i, q_i)$.
In the appendix, we prove the following claim. The proof is by induction on $i$, backwards from $i = \ell$ to $i = 0$. 

\smallskip \noindent
  \qquad If $\dd_{i}$ satisfies all concepts from $\mathit{premise}_i$ in $\interI$, then it also satisfies all concepts from $\mathit{goal}_i$. 

\smallskip
 From this claim, the statement of the lemma follows, as the initial state of $\aA_{\cA}$ is $(\{ \top \}, \cA )$, while the accepting states have their goal already proven.
\end{proof}

%% file: sections/05-completeness.tex

\section{Completeness}\label{sec:completeness}

The goal of this  section is to establish the remaining direction of Theorem~\ref{thm:rewriting}.
But before we can proceed with the proof, we need some additional definitions. 
First, we recall the standard notion of a derivation.

\begin{definition}\label{def:derivation}
Let $\tboxT$ be an $\ourEL$-TBox.
A \emph{$\tboxT$-derivation} of $\cA(\ia)$ from an ABox $\aboxA$ is a finite sequence
$\aboxA_0, \aboxA_1, \ldots, \aboxA_n$ of ABoxes such that $\aboxA_0$
is a subset of $\aboxA$, $\aboxA_n$ contains $\cA(\ia)$, and for every $i < n$
the ABox $\aboxA_{i{+}1}$ can be derived from $\aboxA_i$
(denoted $\aboxA_i \leadsto_{\tboxT} \aboxA_{i{+}1}$), \ie there exists a
GCI $\cB \dlsubseteq \cC$ in $\tboxT$ and an individual name $\ib$ such that:\footnote{For brevity, we abuse the notation and write $\top(\ib) \in \aboxA_i$ to indicate that $\ib$ appears in $\aboxA_i$.}
\begin{itemize}
  \item Both $\cB$ and $\cC$ are concept names or $\top$, $\cB(\ib) \in \aboxA_i$, and $\aboxA_{i{+}1} = \aboxA_i \cup \{ \cC(\ib) \}$;
  \item $\cB$ is of the form $\cB_1 \dland \cB_2$, both $\cB_1(\ib)$ and $\cB_2(\ib)$ belong to $\aboxA_i$, and $\aboxA_{i{+}1} = \aboxA_i \cup \{ \cC(\ib) \}$;
  \item $\cB$ is a concept name or $\top$, $\cC$ is of the form $\exists{\rs}.\cC_1$, $\cB(\ib) \in \aboxA_i$, and $\aboxA_{i{+}1} = \aboxA_i \cup \{ \rs(\ib, \ic), \cC_1(\ic) \}$ for some fresh individual name $\ic$; 
  \item $\cB$ is of the form $\exists{\rs}.\cB_1$, both $\rs(\ib,\ic)$ and $\cB_1(\ic)$ belong to $\aboxA_i$, and $\aboxA_{i{+}1} = \aboxA_i \cup \{ \cC(\ib) \}$; or
  \item $\bot(\ib) \in \aboxA_i$, and $\aboxA_{i+1} = \{ \cA(\ia) \}$. 
\end{itemize}
In all of these cases, if $\cB(\ib)$ is (one of) the mentioned concept assertions in $\aboxA_i$ and $\cA(\ia)$ is the newly added concept assertion to $\aboxA_{i+1}$, we say that $\cB(\ib)$ \emph{was used to infer} $\cA(\ia)$.
We call the number $n$ the \emph{length} of the derivation.
We say that $\cA(\ia)$ is \emph{$\tboxT$-derivable} from $\aboxA$ whenever there
exists a $\tboxT$-derivation of $\cA(\ia)$ from $\aboxA$.~\myqed
\end{definition}

\newcommand{\ABoxAex}{\aboxA_{\mathrm{ex}}}
\newcommand{\TBoxTex}{\tboxT_{\mathrm{ex}}}
\newcommand{\qex}{q_{\mathrm{ex}}}

The following example illustrates the notion of derivations and will be continued later.
\begin{example}\label{example:deriv}
    Consider an ABox $\ABoxAex \deff \{ \cA(\ia) \}$, an $\ourEL$-TBox
    $\TBoxTex \deff \{ \cA \dlsubseteq \cB,\; \cA \dland \cB \dlsubseteq \cC,\;
    \cC \dlsubseteq \exists{\rr}.\top,\; \exists{\rr}.\top \dlsubseteq \cD \}$,
    and an atomic query $\qex \deff \cD(\ia)$.
    We enumerate the GCIs of $\TBoxTex$ in order of appearance as $\alpha_0, \ldots, \alpha_3$.
    The reader may easily verify that the sequence $\aboxA_0, \ldots, \aboxA_4$ is a
    $\TBoxTex$-derivation of $\qex$ from $\ABoxAex$, where $\aboxA_0 = \aboxA$,
    $\aboxA_1 = \aboxA_0 \cup \{ \cB(\ia) \}$,
    $\aboxA_2 = \aboxA_1 \cup \{ \cC(\ia) \}$,
    $\aboxA_3 = \aboxA_2 \cup \{ \rr(\ia, \ib), \top(\ib) \}$,
    and $\aboxA_4 = \aboxA_3 \cup \{ \cD(\ia) \}$.
    In particular, for all $i < 4$ we have $\aboxA_i \leadsto_{\{\alpha_i\}} \aboxA_{i+1}$.\myqed
\end{example}

The following lemma links derivations with entailment. Its proof is standard, and hence we omit it.
\begin{lemma}[folklore]\label{lemma:folklore}
Let $\aboxA$ be an ABox, $\tboxT$ be a stratified $\ourEL$-TBox,
and $\cA(\ia)$ be an instance query.
Then $(\aboxA, \tboxT) \models \cA(\ia)$ if and only if $\cA(\ia)$ can be
$\tboxT$-derived from $\aboxA$.~\myqed
\end{lemma}

Our built-in stratification 
can also be reflected in the derivations, leading to the following notion.  

\begin{definition}
Let $\tboxT$, $\aboxA$, and $\cA(\ia)$ be as in Definition~\ref{def:derivation}, 
and suppose that $\cA$ is of height $\ell$.
A \emph{$(\tboxT, \ell)$-derivation} of $\cA(\ia)$ from $\aboxA$ is  
a finite sequence $(\aboxA_0, \cA_0(\ia_0))$, $(\aboxA_1, \cA_1(\ia_1))$, $\ldots$, $(\aboxA_n, \cA_n(\ia_n))$
of ABoxes and instance queries (called \emph{active queries}) with concepts from $\con{\tboxT\rstr{\ell}}$ such that: 
\begin{itemize}[itemsep=0em]
  \item $\cA_n(\ia_n) = \cA(\ia)$, for all $i \leq n$, the sequence $\aboxA_0, \ldots, \aboxA_i$ is a $\tboxT\rstr{\ell}$-derivation of $\cA_i(\ia_i)$ from $\aboxA$, and
  
  \item for all $i < n$: if $\aboxA_i \leadsto_{\{\cC \dlsubseteq \cD\}} \aboxA_{i+1}$ for some GCI $\cC \dlsubseteq \cD$ from $\tboxT\rstr{\ell}$ with $\cD$ of height $\ell$, then
  \begin{enumerate}
      \item if $\cB'$ is a concept name of height $\ell$ and occurs in $\cC$, then $\cA_i(\ia_i) = \cB'(\ia_i)$; 
      \item $\cA_{i+1}(\ia_{i+1}) = \cB''(\ia_{i+1})$ for the concept name $\cB''$ that occurs in $\cD$. \myqed 
  \end{enumerate}
\end{itemize} 
\end{definition}

The following example illustrates the notion of stratified derivation.
\begin{example}\label{example:stratderiv}
    Let $\ABoxAex$, $\TBoxTex$, and $\qex$ be as in Example~\ref{example:deriv}.
    Suppose that the heights of $\cA$, $\cB$, $\cC$, $\rr$, $\cD$ are $0$, $1$, $2$, $2$, and $3$, respectively.
    The reader may verify that the sequence $(\aboxA_0, \cA_0(\ia_0)), \ldots, (\aboxA_4, \cA_4(\ia_0))$ is a
    $(\TBoxTex, 3)$-derivation of $\qex$ from $\ABoxAex$, where $\cA_0(\ia_0) = \cA(\ia)$,
    $\cA_1(\ia_1) = \cB(\ia)$, $\cA_2(\ia_2) = \cC(\ia)$,
    $\cA_3(\ia_3) = \top(\ib)$, and $\cA_4(\ia_4) = \cD(\ia)$.\myqed
\end{example}

With the next lemma, we connect stratified derivations with the entailment of instance queries.

\begin{lemma}\label{lemma:exists-stratified-derivation}
Let $\tboxT$ be an $\ourEL$-TBox, $\aboxA$ be an ABox, and $\cA(\ia)$ be an 
instance query such that $(\aboxA, \tboxT) \models \cA(\ia)$.
Then there exists a $(\tboxT, \ell)$-derivation of $\cA(\ia)$ from $\aboxA$
with $\ell$ being the height of $\cA$.~\myqed
\end{lemma}
\begin{proof}[Proof Sketch]
Our proof proceeds by induction and our inductive statement is:

\begin{center}
For every $d$, for every $\ell$, and every instance query $\cA(\ia)$ with concept $\cA$ of height~$\ell$ for which there exists a derivation of $(\aboxA, \tboxT) \models \cA(\ia)$ of length $d$,
there also exists a $(\tboxT, \ell)$-derivation of $\cA(\ia)$ from $\aboxA$.
\end{center}

For the case $d = 0$ it suffices to take the pair $(\aboxA, \cA(\ia))$ as the desired $(\tboxT, \ell)$-derivation of $\cA(\ia)$ from $\aboxA$.
For the induction step we can perform a case distinction based on the GCI that was used in the last step of the derivation. We provide a detailed execution of this argument in the appendix.
\end{proof}

We next establish an important property of ``stratified'' derivations.
\begin{lemma}\label{lemma:main-prop-of-stratified-derivation}
Let $\tboxT$ be an $\ourEL$-TBox, $\aboxA$ be an ABox, and let $(\aboxA_0, \cA_0(\ia_0))$, $(\aboxA_1, \cA_1(\ia_1))$, $\ldots$, $(\aboxA_n, \cA_n(\ia_n))$ be a $(\tboxT, \ell)$-derivation of some concept $\cA(\ia)$ from $\aboxA$. For all $i \geq 0$, if $\ia_i \in \indA$ and $\ia_{i-1} \not\in \indA$, such that $\cA_{i-1}(\ia_{i-1})$ was used to infer $\cA_{i}(\ia_{i})$,~then:
\begin{itemize}
  \item if there is no $j < i$ with $\ia_j = \ia_i$, then $(\tboxT\rstr{\ell}, \aboxA_{i-1}\rstr{\ell-1}^{\ia_i}) \models \cA_i(\ia_i)$;
  and
  \item otherwise for some $j < i$ with $\ia_j = \ia_i$ we have $(\tboxT\rstr{\ell}, \aboxA_{i-1}\rstr{\ell-1}^{\ia_i} \cup \{\cA_j(\ia_i)\}) \models \cA_i(\ia_i)$,
\end{itemize}
where $\aboxA_i\rstr{\ell-1}^\ia$ is the restriction of $\aboxA_i$ to assertions about $\ia$ concerning concepts of height at most $\ell-1$.\myqed
\end{lemma}
Informally, in this lemma we claim, that in the described setting, $\cA_i(\ia_i)$ can be inferred solely by considering the already known facts about $\ia_i$ of some lower height, together with the TBox restricted to lower levels, and, if it exists as an active query, a single concept assertion on $\ia_i$ of height $\ell$. 

\begin{proof}
Without loss of generality, assume that $\aboxA$ contains $\top(\ia)$ for all $\ia \in \indA$.
Consider the directed implication graph of $\aboxA_i$ (denoted $\mathcal{G}_{\aboxA_i}$), whose vertices are the concept assertions in $\aboxA_i$. There is a directed edge from $\cB(\ib)$ to $\cC(\ic)$ if there exists $j < i$ such that $\cB(\ib) \notin \aboxA_j$, $\cC(\ic) \in \aboxA_j$, and $\cC(\ic)$ was used to derive $\cB(\ib)$ and add it to $\aboxA_{j+1}$. By Definition~\ref{def:derivation}, if such an edge exists, then either $\ib = \ic$ or $\rs(\ib,\ic) \in \aboxA_{j+1}$.
Inspecting Definition~\ref{def:derivation}, we observe that no derivation step creates edges between existing individuals. Hence, along any path in $\mathcal{G}_{\aboxA_i}$ from an ``anonymous'' individual $\ic \notin \indA$ connected to $\ia_i$ (i.e., $\rs(\ia_i,\ic) \in \aboxA_i$) to a ``named'' individual $\ib \in \indA$, the first named individual encountered is $\ia_i$.
Thus every maximal path $\rho$ in $\mathcal{G}_{\aboxA_i}$ starting at $\cA_i(\ia_i)$ and proceeding to $\cA_{i-1}(\ia_{i-1})$, must have some $k > 2$ such that $\rho_k = \cC(\ia_i)$ for some $\cC$, since the last vertex of $\rho$ lies in $\aboxA$.
Furthermore, at most one such maximal path starting at $\cA_i(\ia_i)$ can contain a vertex $\cA_j(\ia_i)$ of height~$\ell$. Otherwise, suppose two paths $\rho$ and $\rho'$ exist. Then they coincide up to some $\rho_k = \rho'_k = \cB(\ib)$, but diverge at the next step. By the definition of the implication graph, this implies that $\cB(\ib)$ was derived using a rule $\cC \sqcap \cD \dlsubseteq \cB$, with $\ib = \ic = \id$. Then one of $\cC,\cD$ must have strictly lower height, contradicting stratification if a later vertex of height~$\ell$ reappears on the diverging path.
Since only active queries can have height~$\ell$, the assertion $\cA_j(\ia_i)$ must have been active at some stage prior to $i$. All other assertions about $\ia_i$ on paths from $\cA_i(\ia_i)$ have lower height and thus already occur in $\aboxA_{i-1}$.
It remains to derive $\cA_i(\ia_i)$ from $(\tboxT\rstr{\ell}, \aboxA_{i-1}\rstr{\ell-1}^{\ia_i} \cup \{\cA_j(\ia_i)\})$. If no such $\cA_j(\ia_i)$ occurs, we take $\cA_j(\ia_i) = \top(\ia_i)$. Let $S$ be the set of assertions in $\mathcal{G}_{\aboxA_i}$ reachable from $\cA_i(\ia_i)$ without passing through assertions about $\ia_i$. Let $\overline{s} = (s_1,\dots,s_{|S|})$ enumerate $S$ in the order of first occurrence in the derivation. Define:
\[
\aboxA'_0 = \aboxA_{i-1}\rstr{\ell-1}^{\ia_i} \cup \{\cA_j(\ia_i)\},
\]
\[
\aboxA'_{k+1} = \aboxA'_k \cup \{s_{k+1}\} \cup \{\rr(\ib,\ic) \mid \ib,\ic \in \ind{\aboxA'_k \cup \{s_{k+1}\}},\; \rr(\ib,\ic) \in \aboxA_i\}.
\]
We claim this yields a valid derivation of $\cA_i(\ia_i)$.
First, $\overline{s}$ is a reverse topological ordering of $S$: if $(\cB(\ib),\cC(\ic))$ is an edge, then $\cC(\ic)$ was derived before $\cB(\ib)$, hence appears earlier in $\overline{s}$.
Now consider $s_{k+1} = \cB(\ib)$ derived via $\cD \sqsubseteq \cE$. For any predecessor $\cC(\ic)$ of $\cB(\ib)$, we have $\cC(\ic) \in S$ or $\ic = \ia_i$. In the latter case, $\cC(\ia_i)$ is either $\cA_j(\ia_i)$ or already in $\aboxA'_0$. By the ordering of $\overline{s}$, all such $\cC(\ic)$ belong to $\aboxA'_k$.
If $\cD \sqsubseteq \cE$ does not involve existential restrictions, it applies directly. If it is of the form $\exists \rr.\cC \sqsubseteq \cB$, then all required role assertions are present in $\aboxA'_k$ by construction. If it is of the form $\cC \sqsubseteq \exists \rr.\cB$, then $\ib$ is introduced as a fresh individual, and $\aboxA'_{k+1}$ correctly reflects this extension.
Thus, each step is valid, completing the derivation.
\end{proof}

\bigskip 

We can now demonstrate how to retrieve accepting runs of the nested automata from the stratified derivations guaranteed by the entailment of instance queries.
In this way, we can finally establish the remaining direction of Theorem~\ref{thm:rewriting}, restated below for convenience.

\begin{lemma}\label{lemma:rewriting-completeness}
  Let $\tboxT$ be a stratified $\ourEL$-TBox,
  $\cA(\ia)$ be an instance query with $\cA$ of height $n$, and
  $\aA_{\cA}$ be the $n$-nested NFA for $\cA$ constructed in Definition~\ref{def:automaton-general-case}.
  Then for all ABoxes $\aboxA$ we have that
  $(\aboxA, \tboxT) \models \cA(\ia)$ implies
  $\aboxA \models \exists{x}\, \aA_{\cA}(\ia, x)$, \ie there exists a accepting
  run of $\aA_{\cA}$ starting from $\ia^{\interI}$.~\myqed
\end{lemma}
\begin{proof}[Proof Sketch]
    Fix $\tboxT$ and $\aboxA$ as in the statement of the lemma, and let $\interI \models (\tboxT, \aboxA)$. We
    proceed by induction over the length $d$ of the run. More specifically, we prove the following inductive claim: 
    
\medskip \noindent
    For all $d \in \N$, all $\ell \in \N$, all $\ia \in \indA$, 
    and all concept names $\cA \in \con{\tboxT\rstr{\ell}}$,
    whenever there is a $(\tboxT, \ell)$-derivation of $\cA(\ia)$ from $\aboxA$ of length $d$, then there is an accepting run of $\aA_{\cA}$ starting from $\ia^\interI$.
    
    \medskip 
    The base case is immediate. Indeed, take any $\ell$ and $\cA \in \con{\tboxT\rstr{\ell}}$, and suppose that there is a $(\tboxT, \ell)$-derivation of $\cA(\ia)$ from $\aboxA$ of length $d$.
    Then such a derivation has the form $(\aboxA_0, \cA(\ia))$ for some $\aboxA_0 \subseteq \aboxA$ with $\cA(\ia) \in \aboxA_0$. 
    Thus, to construct an accepting run of $\aA_\cA$ we can step from $q_0 \deff (\{ \top \}, \cA )$ to $q_1 \deff (\{ \cA, \top \},  \cA )$, via a \ref{trans:data}-transition, which is already a final state. 
    More formally, $(\ia^{\interI}, q_0, \cA?, q_1, \ia^{\interI})$ is the desired run of $\aA_{\cA}$ starting from $\ia^{\interI}$. 

    For the inductive step, we again do a case distinction on the last step of the $(\tboxT, \ell)$-derivation to simulate that derivation step in the automaton. In particular ``reasoning in the anonymous part'' is possible thanks to Lemma \ref{lemma:main-prop-of-stratified-derivation}. In the appendix we provide a full version of this argument.
\end{proof}


The following example illustrates the proof above.
\begin{example}\label{example:compl-proof}
    Let $\ABoxAex$, $\TBoxTex$, and $\qex$ be as in Example~\ref{example:deriv}, and let
    $(\aboxA_0, \cA_0(\ia_0)), \ldots, (\aboxA_4, \cA_4(\ia_0))$ be the
    $(\TBoxTex, 3)$-derivation of $\qex$ from $\ABoxAex$ from Example~\ref{example:stratderiv}.
    Fix a model $\interI \models (\TBoxTex, \ABoxAex)$.
    By inspecting the proof of Lemma~\ref{lemma:rewriting-completeness}, we see that the following
    run $(\ia^{\interI}, \delta_0, \ia^{\interI}), \ldots, (\ia^{\interI}, \delta_4, \ia^{\interI})$ of $\aA_{\cD}$ starting from $\ia^{\interI}$ is produced, where:
    \begin{itemize}
        \item $\delta_0 = ((\{ \top \}, \cD), \top?, (\{ \top \}, \cC))$ corresponds to the use of a \ref{trans:anon}-transition parametrised by $\cC$, invoking the anonymous part of $\interI$.
        \item $\delta_1 = ((\{ \top \},\cC ), \cA?, (\{ \top, \cA \}, \cC ))$ corresponds to the use of a \ref{trans:data}-transition reading $\cA(\ia)$ from the ABox.
        \item $\delta_2 = ((\{ \top, \cA \}, \cC ), \top?, (\{ \top, \cA\}, \cB ))$ corresponds to the use of a \ref{trans:noc}-transition invoking the GCI $\cA \dland \cB \dlsubseteq \cC$.
        \item $\delta_3 = ((\{ \top, \cA \}, \cB ), \top?, (\{ \top, \cA\}, \cA ))$ corresponds to the use of a \ref{trans:sbus}-transition invoking the GCI $\cA  \dlsubseteq \cB$. This puts the automaton in an accepting state.~\myqed
    \end{itemize}
\end{example}

We have proved that our rewriting of instance queries under $\ourEL$-ontologies into N2RPQs is sound and complete. As evaluating N2RPQs is feasible in $\NL$ \cite{DBLP:journals/mst/ReutterRV17}, we obtain the desired complexity result. 

\begin{corollary}\label{crl:nl}
  Take a fixed $\ourEL$-TBox $\tboxT$ and an instance query $q$.
  For every ABox $\aboxA$, deciding $(\aboxA, \tboxT) \models q(\ia)$ is feasible in $\NL$.
\end{corollary}

%% file: sections/06-combined-complexity.tex

\section{Combined Complexity}\label{sec:combined-complexity}

We conclude the paper by establishing $\PSpace$-completeness of instance query entailment in $\ourEL$.

\paragraph*{Lower bound.}
We reduce from QBF, which is $\PSpace$-hard even for prenex 3-DNF formulas~\cite[Thm.~6.1]{Schaefer78}. The reduction already works for linear orders $\preceq$.
Consider $\varphi = Q_1 x_1 \cdots Q_n x_n \psi$, where $Q_i \in \{\exists,\forall\}$ and $\psi = \bigvee_{j=1}^m \lambda_{j1} \land \lambda_{j2} \land \lambda_{j3}$, where $\lambda$s are literals. 
We define a singleton ABox $\aboxA$, a $\tboxT$, and a concept $\cC_0^{\text{True}}$ such that $\cC_0$ holds for the only individual in $\aboxA$ if and only if $\varphi$ is valid.
We employ the following concept~ and role~names
\[
\{\cL_i,\cX_i^0,\cX_i^1,\cC_i^{\text{True}},\cC_i^{\text{True},0},\cC_i^{\text{True},1} \mid 0 \leq i \leq n\}
\cup \{\cA_j^1,\cA_j^2 \mid 1 \leq j \leq m\}
\qquad 
\{\rr_{0,i},\rr_{1,i} \mid 0 \leq i \leq n\}.
\]
Intended models are binary trees of depth $n$ rooted at $\ia$, where level $i$ (marked by $\cL_i$) encodes the value of $x_i$ via $\cX_i^0$ or $\cX_i^1$. Parent-to-child connections on level $i$ are realised by roles $\rr_{d,i}$ (where $d \in \{0,1\}$ denotes a direction). The concepts $\cC_i^{\text{True}}$ propagate truth of the suffix formula, and $\cA_j^1,\cA_j^2$ simulate conjunctions of literals.
Let $\aboxA = \{\cL_0(\ia)\}$ and query $\cC_0^{\text{True}}(\ia)$. The TBox $\tboxT$ is defined as:
    \begin{align}
        \tboxT &\coloneqq \{\cL_i \sqsubseteq \exists \rr_{0,i+1}.\cL_{i+1},\cL_i \sqsubseteq \exists \rr_{1,i+1}.\cL_{i+1} \mid i \in \{0,\ldots,n-1\}\}\\
        & \cup \{\exists \overline{\rr}_{0,i+1}.\cL_i \sqsubseteq  \cX_{i+1}^0, \exists \overline{\rr}_{1,i+1}.\cL_i \sqsubseteq  \cX_{i+1}^1 \mid i \in \{0,\ldots,n-1\}\}\\
        & \cup \{\exists \overline{\rr}_{d,i}.\cX_i^0 \sqsubseteq  \cX_{i}^0, \exists \overline{\rr}_{d,i}.\cX_i^1 \sqsubseteq  \cX_{i}^1 \mid i \in \{  1, \ldots, n \}, d \in \{0,1\}\}\\
        & \cup \{\cL_n \sqcap \cX_i^1 \sqsubseteq \cA_j^1 \mid \lambda_{j1} =  x_i, j \in \{ 1, \ldots, m \}, i \in \{ 1, \ldots, n \}\}\\
        & \cup \{\cL_n \sqcap \cX_i^0 \sqsubseteq \cA_j^1 \mid \lambda_{j1} = \neg x_i, j \in \{ 1, \ldots, m \}, i \in \{ 1, \ldots, n \}\}\\
        & \cup \{\cA_j^1 \sqcap \cX_i^1 \sqsubseteq \cA_j^2 \mid \lambda_{j2} =  x_i, j \in \{ 1, \ldots, m \}, i \in \{ 1, \ldots, n \}\}\\
        & \cup \{\cA_j^1 \sqcap \cX_i^0 \sqsubseteq \cA_j^2 \mid \lambda_{j2} = \neg x_i, j \in \{ 1, \ldots, m \}, i \in \{ 1, \ldots, n \}\}\\
        & \cup \{\cA_j^2 \sqcap \cX_i^1 \sqsubseteq \cC^{\text{True}}_n \mid \lambda_{j3} =  x_i, j \in \{ 1, \ldots, m \}, i \in \{ 1, \ldots, n \}\}\\
        & \cup \{\cA_j^2 \sqcap \cX_i^0 \sqsubseteq \cC^{\text{True}}_n \mid \lambda_{j3} = \neg x_i, j \in \{ 1, \ldots, m \}, i \in \{ 1, \ldots, n \}\}\\
        & \cup \{\exists r_0. \cC_{i+1}^{\text{True}} \sqsubseteq \cC_{i}^{\text{True},0}, \exists r_1. \cC_{i+1}^{\text{True}} \sqsubseteq \cC_{i}^{\text{True},1}\mid i \in \{0,\ldots,n-1\}\}\\
        & \cup \{\cC_{i}^{\text{True},0} \sqsubseteq \cC_{i}^{\text{True}}, \cC_{i}^{\text{True},1} \sqsubseteq \cC_{i}^{\text{True}} \mid i \in \{0,\ldots,n-1\}, Q_{i+1} = \exists\}\\
        & \cup \{\cC_{i}^{\text{True},0} \sqcap \cC_{i}^{\text{True},1} \sqsubseteq \cC_{i}^{\text{True}} \mid i \in \{0,\ldots,n-1\}, Q_{i+1} = \forall\}
    \end{align}
    
    Finally, we  can order our subset of $\Clang \cup \Rlang$ with $\preceq$ as follows:
    \begin{align*}
        & \cL_0 \prec \rr_{0,1} \prec \rr_{1,1} \prec \cX^0_1 \prec \cX^1_1 \prec \cL_1 \prec \rr_{0,2} \prec \rr_{1,2} \prec \cX^0_2 \prec \cX^1_2 \prec  \ldots \cL_n \prec\\
        & \cA^1_1 \prec \cA^2_1 \prec \cA^1_2 \prec \cA^2_2 \prec \ldots \cA^1_m \prec \cA^2_m \prec\\
        & \cC_n^\text{True} \prec \cC_{n}^{\text{True},0} \prec \cC_{n}^{\text{True},1} \prec \cC_{n-1}^\text{True} \prec \cC_{n-1}^{\text{True},0} \prec \cC_{n-1}^{\text{True},1} \prec \ldots \prec \cC_0^\text{True}
    \end{align*}

The proof of the following lemma is standard. Check the appendix for details.
\begin{lemma}
$(\aboxA,\tboxT) \models \cC_0^{\text{True}}(\ia)$ if and only if $\varphi$ is valid. Hence, instance query entailment in $\ourEL$ is $\PSpace$-hard (in combined complexity).\myqed
\end{lemma}

\paragraph*{Upper bound.}
Let $\aboxA$ be an input ABox, $\tboxT$ an $\ourEL$-TBox, and $q \deff \cA(\ia)$ an instance query. By Theorem~\ref{thm:rewriting}, to decide whether $(\tboxT,\aboxA) \models q$, it suffices to construct the nNFA $\aA$ given by the rewriting and check for an accepting run of $\aA$ on $\aboxA$ starting at $\ia$.
The size of $\aA$ is exponential in $|\tboxT|$ (due to subsets of concepts in the premise component of states), so it cannot be materialised explicitly. Instead, we construct it on the fly during evaluation. Any accepting run has length bounded by the number of states of $\aA$, hence exponential in $|\tboxT|$. Such a bound can be tracked using an exponentially large counter encoded in polynomial space.
The procedure guesses a run of $\aA$ on $\aboxA$ from $\ia$ step by step, incrementing the counter and verifying the correctness of each transition. If an accepting state is reached within the bound, we accept; otherwise, we reject. Since only the current state and the counter need to be stored, the procedure runs in polynomial space. Thus, instance query entailment in $\ourEL$ is in $\PSpace$ (in combined complexity).
The only subtle point is the verification of transitions. While most transitions of $\aA$ can be checked directly using the ABox and the current state, Steps~\ref{trans:anon} and~\ref{trans:aut} require additional care.
For \ref{trans:aut} transitions, observe that the automaton test $\aA_\cB?$ involves only concepts of strictly lower height than $\cA$. Thus, by an inductive construction over heights and by the use of $\PSpace$ oracles, this step can be fully implemented in polynomial space.
For \ref{trans:anon} transitions, which correspond to classical reasoning in $\ourEL$, we adapt the standard alternating polynomial-space algorithm for $\ELIHbot$. The key observation is that the number of alternations is bounded by the height of the input concept. By the classical result that alternating $\PSpace$ with polynomially many alternations coincides with $\PSpace$~\cite[Thm.~4.2]{ChandraKS81}, this step is also feasible in polynomial space.
A more detailed exposition will be given in revised version of the paper.
We thus conclude:

\begin{theorem}
Instance query entailment in $\ourEL$ is $\PSpace$-complete (in combined complexity).~\myqed
\end{theorem}

%% file: sections/07-conclusions.tex

\section{Conclusions and Future Work}\label{sec:conclusions}

In this work, we introduced $\ourEL$, a novel fragment of $\ELIHbot$ that extends $\DLLite$ and is based on the idea of ``stratified conjunction.'' By means of a suitable rewriting into nested two-way regular path queries ($\NTwoRPQ$s), we established that the instance query entailment problem for our logic is $\NL$-complete with respect to data complexity.

Our future and ongoing work will pursue several directions. First, we plan to incorporate role hierarchies into our logic, which should be achievable with only minor modifications to our approach. Second, we aim to move beyond instance queries and investigate the combined complexity of query answering for $\ourEL$ with respect to the full class of nested C2RPQs. We are optimistic that this can be achieved through a suitable adaptation of the techniques presented by Bienvenu et al.~\cite{BienvenuCOS14}. Third, we intend to provide support for concrete domains. To do so, we plan to combine the query rewriting approach of Baader et al.~\cite{BaaderBL17} with the recent model of C2RPQs with data values introduced by Figueira et al.~\cite{FigueiraJL22}. Finally, since the target of our rewriting is (a fragment of) GQL, we plan to investigate the practical applicability of our approach by implementing it and evaluating it on real-world datasets.

%% file: sections/appendix-soundness.tex
\section{Appendix on Soundness}
Here we provide the proof of Lemma \ref{lemma:rewriting-soundness}, which we restate once more below.
\begin{lemma}
  Let $\tboxT$ be a stratified $\ourEL$-TBox, 
  $\cA(a)$ be an instance query, and
  $\aA_{\cA}$ be an nNFA from Definition~\ref{def:automaton-general-case}.
  Then for all for all ABoxes $\aboxA$ we have that 
  $\aboxA \models \exists{x} \aA_{\cA}(\ia, x)$ implies 
  $(\aboxA, \tboxT) \models \cA(\ia)$.~\myqed
\end{lemma}
\begin{proof}
Fix a TBox $\tboxT$. 
The proof is by induction on
the height $n$ of the concept name $A$
(the induction in turn contains a nested induction on the length of a run of the automaton).
We assume that for all $n' < n$ and concept names $\cB \in \con{\tboxT\rstr{n'}}$, the statement of Lemma~\ref{lemma:rewriting-soundness} holds for the instance query $\cB(\ia)$. 

Let us move forward to an instance query $\cA(\ia)$ with $\cA \in \con{\tboxT\rstr{n}}$.
Fix an ABox $\aboxA$, a model $\interI$ of $(\aboxA, \tboxT)$, and let $\rho = \rho_1 \cdots \rho_\ell$ be an accepting run of $\aA_{\cA}$ in $\interI$ starting from $\ia^{\interI}$, which exists by entailment $\aboxA \models \exists{x} \aA_{\cA}(\ia, x)$).
For convenience, let $\rho_i$ unfold to $(\dd_{i-1}, \delta_i, \dd_i)$ with $q_i = (\mathit{premise}_i, \mathit{goal}_i)$ and $\delta_i = (q_{i-1}, w_i, q_i)$.
We now proceed by induction on $i$, backwards from $i = \ell$ to $i = 0$, and  prove that:\\ 
$(\dagger)$ if $\dd_{i}$ satisfies all concepts from $\mathit{premise}_i$ in $\interI$, then it also satisfies all concepts from $\mathit{goal}_i$ in $\interI$.

Then the statement of the lemma follows as the initial state of $\aA_{\cA}$ is $(\{ \top \}, \cA )$, while the accepting states have their goal already proven.

The base case follows trivially, since $\mathit{goal}_\ell \in \mathit{premise}$ or $\bot \in \mathit{premise}$.
Hence, let us focus on the case when $i < \ell$ and the induction hypothesis already holds for all $i' > i$.
Suppose $\dd_i$ satisfies all concepts from $\mathit{premise}_i$ in $\interI$.
We distinguish cases depending on how the transition $\delta_i \deff (q_{i-1}, w_i, q_i)$ was created in the construction of $\aA_{\cA}$.
\begin{itemize}[itemsep=0em]

  \item It was created by the rule~\ref{trans:weak}.
  After inspecting the construction, we see that:
  (i) $\mathit{premise}_{i+1}$ is a subset of $\mathit{premise}_{i}$,
  (ii) $\mathit{goal}_{i+1} = \mathit{goal}_{i}$,
  (iii) $w_i = \top?$, and thus
  (iv) $\dd_i = \dd_{i+1}$.
  As $\dd_i$ satisfies all concepts from $\mathit{premise}_{i}$, then by (iv) and (i) we have that $\dd_{i+1}$ satisfies all concepts from $\mathit{premise}_{i+1}$.
  Thus, we can invoke the inductive assumption for $i+1$ to infer that $\dd_{i+1}$ satisfies $\mathit{goal}_{i+1}$ in $\interI$.
  Together with (iv) and (ii), this implies that $\dd_i$ satisfies $\mathit{goal}_i$ in $\interI$, as desired.

  \item It was created by Rule~\ref{trans:data}.
  After inspecting the construction, we see that:
  (i) $\mathit{goal}_{i+1} = \mathit{goal}_{i}$,
  and there exists a concept $\cC$ for which:
  (ii) $\mathit{premise}_{i+1} = \mathit{premise}_{i} \cup \{ \cC \}$,
  (iii) $w_i = \cC?$, and hence
  (iv) $\dd_i = \dd_{i+1}$ and
  (v) $\dd_i \in \cC^{\interI}$.
  By our assumptions together with (ii) and (v), we have that $\dd_i$ satisfies all concepts from $\mathit{premise}_{i+1}$.
  By (iv), we can invoke the inductive hypothesis for $i+1$ to infer that $\dd_i$ satisfies all concepts from $\mathit{goal}_{i+1}$ in $\interI$, which by (i) concludes the proof.

  \item It was created by the rule~\ref{trans:sbus}. 
  After inspecting the construction, we see that there is a GCI $\cC \dlsubseteq \cD$ in $\tboxT$ such that:
  (i) $ \mathit{goal}_{i} = \cD$,
  (ii) $\mathit{premise}_{i+1} = \mathit{premise}_i$,
  (iii) $\mathit{goal}_{i+1} = \cC$,
  (iv) $w_i = \top?$, and thus
  (v) $\dd_i = \dd_{i+1}$.
  By our assumptions together with (ii) and (v), we can invoke the inductive hypothesis for $i{+}1$ to infer that $\dd_i$ satisfies $\mathit{goal}_{i+1} = \cC$ in $\interI$.
  It remains to show that $\dd_i$ also satisfies $\cD$. This follows from the fact that $\dd_i$ satisfies~$\cC$ and $\interI$ is a model of $\tboxT$, containing $\cC \dlsubseteq \cD$.
    
  \item It was created by Rule~\ref{trans:succ}.
  After inspecting the construction, we see that there is a GCI $\exists{\rs}.\cC \dlsubseteq \cD$ in $\tboxT$ such that:
  (i) $\mathit{goal}_i =  \cD $,
  (ii) $\mathit{goal}_{i+1} =  \cC $,
  (iii) $\mathit{premise}_{i+1} = \{ \top \}$,
  (iv) $w_i = \rs$, and thus
  (v) $(\dd_i, \dd_{i+1}) \in \rs^{\interI}$.
  We can trivially conclude that $\dd_{i+1}$ satisfies all concepts from $\mathit{premise}_{i+1}$. Invoking the inductive hypothesis, we infer that $\dd_{i+1}$ also satisfies $\mathit{goal}_{i+1} = \cC$ by (ii).
  From (v) and the fact that $\interI$ is a model of $\exists{\rs}.\cC \dlsubseteq \cD$, we conclude that $\dd_i$ satisfies $\cD$, which by (i) concludes the proof.

  \item If $\delta_i$ was created by Rule~\ref{trans:anon}.
  After inspecting the construction, we see that there are a concepts~$\cD$ and $\cB$ from $\tboxT$ such that: 
  (i) $\mathit{goal}_{i+1} = \cB$,
  (ii) $\mathit{premise}_{i+1} = \mathit{premise}_i$,
  (iii) $(\{ \cC(\ia) \mid \cC \in \mathit{premise}_i \} \cup \{\cB(\ia)\}, \tboxT)$ entails $\cD(\ia)$,
  (iv) $w_i = \top?$, and thus
  (v) $\dd_i = \dd_{i+1}$.
  By our assumptions, $\dd_i$ satisfies all concepts from $\mathit{premise}_i$, and hence, by (iii) and $\interI \models \tboxT$, it also satisfies $\cD$, if $\dd_i$ satisfies $\cB$.
  Together with (ii) and (v), we can invoke the inductive hypothesis for $i+1$ to infer that $\dd_i$ satisfies all concepts from $\mathit{goal}_{i+1}$ in $\interI$, which by (i) concludes the proof.

  \item It was created by the rule~\ref{trans:noc}. 
  After inspecting the construction, we see that there is a GCI $\cB \dland \cC \dlsubseteq \cD$ in $\tboxT$ such that:
  (i) $\cB \in \mathit{premise}_{i}$,
  (ii) $\mathit{goal}_{i} = \cD$,
  (iii) $\mathit{premise}_{i+1} = \mathit{premise}_i$,
  (iv) $\mathit{goal}_{i+1} = \cC $,
  (v) $w_i = \top?$, and thus
  (vi) $\dd_i = \dd_{i+1}$.
  By our assumptions together with (iii) and (vi), we can invoke the inductive hypothesis for $i{+}1$ to infer that $\dd_i$ satisfies $\mathit{goal}_{i+1} = \cC$ in $\interI$.
  We also know that $\dd_i$ satisfies $\cB$ by (i). Since $\interI$ is a model of $\tboxT$, it follows that $\dd_i$ satisfies $\cD$ as well, which by (iv) concludes the proof.

  \item If $\delta_i$ was created by Rule~\ref{trans:aut}.
  After inspecting the construction, we see that there is a concept $\cB \in \con{\tboxT\rstr{n-1}}$ such that:
  (i) $\mathit{goal}_{i+1} = \mathit{goal}_{i}$,
  (ii) $\mathit{premise}_{i+1} = \mathit{premise}_i \cup \{ \cB \}$,
  (iii) $w_i = \aA_{\cB}?$, and thus
  (iv) $\dd_i = \dd_{i+1}$.
  It suffices to show that $\dd_i$ satisfies $\cB$ in $\interI$. Indeed, by (ii) and (iv), we can invoke the inductive hypothesis for $i+1$ to infer that $\dd_i$ satisfies all concepts from $\mathit{goal}_{i+1}$ in $\interI$, which by (i) concludes the proof.
  To see that $\dd_i$ satisfies $\cB$, we invoke the inductive hypothesis for $i+1$ and $\cB$ (recall that $\cB$ has height strictly smaller than $n$).

  \end{itemize}
  As we exhausted all cases, the proof is complete.
\end{proof}

%% file: sections/appendix-completeness.tex
\section{Appendix on Completeness}

\begin{proof}[Proof of Lemma \ref{lemma:exists-stratified-derivation}]
Our proof proceeds by induction and our inductive statement is:

\begin{center}
For every $d$, for every $\ell$, and every instance query $\cA(\ia)$ with concept $\cA$ of height~$\ell$ for which there exists a derivation of $(\aboxA, \tboxT) \models \cA(\ia)$ of length $d$,
there also exists a $(\tboxT, \ell)$-derivation of $\cA(\ia)$ from $\aboxA$.
\end{center}

The case $d = 0$ is immediate: it suffices to take the pair $(\aboxA, \cA(\ia))$ as the desired $(\tboxT, \ell)$-derivation of $\cA(\ia)$ from $\aboxA$.
We now assume $d > 0$ and that the inductive hypothesis holds for all smaller values of $d$.
Next, take any $\ell$ and a concept $\cA$ of height $\ell$, and let $\aboxA_0, \ldots, \aboxA_d$ be the shortest $\tboxT$-derivation of $\cA(\ia)$ from $\aboxA$.
Aiming to construct the desired $(\tboxT, \ell)$-derivation of $\cA(\ia)$ from $\aboxA$, we proceed by case distinction on the shape of the GCI $\cB \dlsubseteq \cA$ such that $\aboxA_{d-1} \leadsto_{\{ \cB \dlsubseteq \cA \}} \aboxA_{d}$. 
\begin{enumerate}
\item $\cB = \top$ and $\cA$ is a concept name. Then, exactly as in the base case, we simply take $(\aboxA, \cA(\ia))$ as the desired derivation.
\item $\cA$ and $\cB$ are concept names with $\cB \preceq \cA$.
Observe that $\aboxA_0, \ldots, \aboxA_{d-1}$ is a $\tboxT$-derivation of $\cB(\ia)$ from $\aboxA$ of length $d{-}1$. By the inductive hypothesis, there exists a $(\tboxT, \ell)$-derivation $(\aboxA_0', \cA_0'(\ia_0)), \ldots, (\aboxA_{k}', \cA_k'(\ia_k))$ of $\cB(\ia)$ from $\aboxA$.
By design, $\ia_k = \ia$, $\cA_k' =\cB$, and $\aboxA_0' \subseteq \aboxA$.
Thus $(\aboxA_0', \cA_0'(\ia_0)), \ldots, (\aboxA_{k}', \cA_k'(\ia)), (\aboxA_{k}' \cup \{ \cA(\ia) \}, \cA(\ia))$ is the desired $(\tboxT, \ell)$-derivation of $\cA(\ia)$ from $\aboxA$.
\item $\cA$ is a concept name (or $\cA = \bot$) and $\cB = \cB_1 \dland \cB_2$ for concept names $\cB_1 \preceq \cA$ and $\cB_2 \prec \cA$ (and analogously for $\cB = \cB_2 \dland \cB_1$).
Observe that $\aboxA_0, \ldots, \aboxA_{d-1}$ is a $\tboxT$-derivation of both $\cB_1(\ia)$ and $\cB_2(\ia)$ from $\aboxA$ of length $d{-}1$.
By the inductive hypothesis, there exist $(\tboxT, \ell)$- and $(\tboxT, \ell{-}1)$-derivations
$(\aboxA_0', \cA_0'(\ia_0')), \ldots, (\aboxA_{k}', \cA_k'(a_k'))$ and
$(\aboxA_0'', \cA_0''(\ia_0'')), \ldots, (\aboxA_{m}'', \cA_m''(\ia_m''))$ of
$\cB_1(\ia)$ and $\cB_2(\ia)$ from $\aboxA$, respectively.
Note that $\cA_k' = \cB_1$ and $\ia = \ia'_k$. One can verify that the sequence
$(\aboxA_0', \cA_0'(\ia_0')), \ldots, (\aboxA_{k}', \cA_k'(\ia)), (\aboxA_{k}' \cup \aboxA_0'', \cB_1(\ia)), \ldots, (\aboxA_{k}' \cup \aboxA_{m}'', \cB_1(\ia)), (\aboxA_{k}' \cup \aboxA_{m}'' \cup \{ \cA(\ia) \}, \cA(\ia))$
is the desired $(\tboxT, \ell)$-derivation of $\cA(\ia)$ from $\aboxA$.
\item $\cB = \top$ and $\cA = \exists{\rs}.\top$.
We simply take $(\aboxA, \top(\ia)), (\aboxA \cup \{ \rs(\ia,\ib) \}, \top(\ib))$ as the desired derivation, where $\ib$ is a fresh individual name.
\item $\cB = \top$ and $\cA = \exists{\rs}.\cA'$ for some concept name $\cA'$ with $\cB \preceq \cA'$.
We simply take $(\aboxA, \top(\ia)), (\aboxA \cup \{ \rs(\ia,\ib), \cA'(\ib) \}, \cA'(\ib))$ as the desired derivation, where $\ib$ is a fresh individual name.
\item $\cB$ is a concept name and $\cA = \exists{\rs}.\cA'$ for some concept name $\cA'$ with $\cB \preceq \cA'$.
Then there is a $\tboxT$-derivation $\aboxA_0, \ldots, \aboxA_{d-1}$ of $\cA'(\ia)$ from $\aboxA$ of length $d{-}1$. By the induction hypothesis there exists a $(\tboxT, \ell)$-derivation $(\aboxA_0', \cA_0'(\ia_0)), \ldots, (\aboxA_{k}', \cA_k'(a_k))$ of $\cA'(\ia)$ from $\aboxA$, where $\cA'_k = \cA' = \cB$ or $\cB \prec \cA'$.
In either case, we can take $(\aboxA_0', \cA_0'(\ia_0)), \ldots, (\aboxA_{k}', \cA_k'(a_k)), (\aboxA_k' \cup \{ \rs(\ia,\ib), \cA'(\ib) \}, \cA'(\ib))$ as the desired derivation, where $\ib$ is a fresh individual name.
\item $\cA$ is a concept name (or $\cA = \bot$) and $\cB = \exists{\rs}.\cC$ for $\cC \in \{\top, \cA\}$ or $\cC$ another concept name. Then there is a $\tboxT$-derivation $\aboxA_0, \ldots, \aboxA_{d-1}$ of $\cC(\ib)$ from $\aboxA$ of length $d{-}1$ with $\rs(\ia,\ib) \in \aboxA_{d-1}$. By the induction hypothesis there exists a $(\tboxT, \ell)$-derivation $(\aboxA_0', \cA_0'(\ia_0)), \ldots, (\aboxA_{k}', \cA_k'(a_k))$ of $\cC(\ib)$ from $\aboxA$.
By design, $\ia_k = \ib$, $\cA_k' =\cC$, and $\aboxA_0' \subseteq \aboxA$.
Thus $(\aboxA_0', \cA_0'(\ia_0)), \ldots, (\aboxA_{k}', \cA_k'(\ia_k)), (\aboxA_{k}' \cup \{ \cA(\ia) \}, \cA(\ia))$ is the desired $(\tboxT, \ell)$-derivation of $\cA(\ia)$ from $\aboxA$. 
\end{enumerate}
As we exhausted all the cases, we are done with the inductive step. This concludes the proof.
\end{proof}

We now also give a full version of the proof for Lemma \ref{lemma:rewriting-completeness} restated below.
\begin{lemma}
  Let $\tboxT$ be a stratified $\ourEL$-TBox,
  $\cA(\ia)$ be an instance query with $\cA$ of height $n$, and
  $\aA_{\cA}$ be the $n$-nested NFA for $\cA$ constructed in Definition~\ref{def:automaton-general-case}.
  Then for all ABoxes $\aboxA$ we have that
  $(\aboxA, \tboxT) \models \cA(\ia)$ implies
  $\aboxA \models \exists{x}\, \aA_{\cA}(\ia, x)$, \ie there exists a accepting
  run of $\aA_{\cA}$ starting from $\ia^{\interI}$.~\myqed
\end{lemma}
\begin{proof}[Proof of Lemma~\ref{lemma:rewriting-completeness}]
    Fix $\tboxT$ and $\aboxA$ as in the statement of the lemma, and let $\interI \models (\tboxT, \aboxA)$. We proceed by induction over $d$, with the inductive claim stated below.
    \begin{center}
    For all $d \in \N$, all $\ell \in \N$, all $\ia \in \indA$, 
    and all concept names $\cA \in \con{\tboxT\rstr{\ell}}$,
    whenever there is a $(\tboxT, \ell)$-derivation of $\cA(\ia)$ from $\aboxA$ of length $d$, then there is an accepting run of $\aA_{\cA}$ starting from $\ia^\interI$.
    \end{center}

    The base case is immediate. Indeed, take any $\ell$ and $\cA \in \con{\tboxT\rstr{\ell}}$, and suppose that there is a $(\tboxT, \ell)$-derivation of $\cA(\ia)$ from $\aboxA$ of length $d$.
    Then such a derivation has the form $(\aboxA_0, \cA(\ia))$ for some $\aboxA_0 \subseteq \aboxA$ with $\cA(\ia) \in \aboxA_0$. 
    Thus, to construct an accepting run of $\aA_\cA$ we can step from $q_0 \deff (\{ \top \}, \cA )$ to $q_1 \deff (\{ \cA, \top \},  \cA )$, via a \ref{trans:data}-transition, which is already a final state. 
    More formally, $(\ia^{\interI}, q_0, \cA?, q_1, \ia^{\interI})$ is the desired run of $\aA_{\cA}$ starting from $\ia^{\interI}$. 
    
    We now move to the inductive step, assuming that our inductive statement holds for all $d' < d$. 
    Consider any $\ell$, an $\ia$, and a concept name $\cA \in \con{\tboxT\rstr{\ell}}$. Suppose that $(\aboxA_0, \cA_0(\ia_0))$, $(\aboxA_1, \cA_1(\ia_1))$, $\ldots$, $(\aboxA_d, \cA_d(\ia_d))$ (where $\cA_d = \cA$ and $\ia_d = \ia$) is a $(\tboxT, \ell)$-derivation of $\cA(\ia)$ from $\aboxA$ of length $d$.
    We~perform a case analysis on the form of the GCI $\cB \dlsubseteq \cA$ for which $\aboxA_{d-1} \leadsto_{\{ \cB \dlsubseteq \cA \}} \aboxA_{d}$ holds, in each cases pinpointing how to construct an accepting run of $\aA_{\cA}$ starting from $\ia^{\interI}$. Correctness of our construction will be immediate, relying on the inductive hypothesis and the design of our~automata.

    \begin{enumerate}
        
        \item Suppose that $\cA$ is a concept name and $\cB = \top$. 
        To construct an accepting run of $\aA_\cA$ we can step from $q_0 \deff (\{ \top \}, \cA)$ to $q_1 \deff (\{ \top \}, \top)$, via a \ref{trans:sbus}-transition, which is already a final state. 
        Then the sequence $(\ia^{\interI}, q_0, \top?, q_1, \ia^{\interI})$ is clearly a run of $\aA_{\cA}$ starting from $\ia^{\interI}$.

        
        \item Suppose $\cA$ and $\cB$ are concept names with $\cB \preceq \cA$. Then, our derivation restricted to the first $d-1$ elements is a $(\tboxT, \ell)$-derivation of $\cB(\ia)$ from $\aboxA$ of length $d-1$.
        Hence, we can invoke the inductive hypothesis to obtain an accepting run $\rho \deff (\ia^{\interI}, \delta_1, \dd_2), \ldots, (\dd_{k-1}, \delta_{k}, \dd_{k})$ of $\aA_{\cB}$ starting from~$\ia^{\interI}$. 
        Note that this is also a run of $\cA$ due to $\cB \preceq \cA$.
        Applying the \ref{trans:sbus}-transition from $(\{\top\},\cA)$ to $(\{\top\},\cB)$ and then following run $\rho$ from there yields the desired run $(\ia^{\interI}, (\{\top\},\cA), \top?, (\{\top\},\cB), \ia^{\interI}), \rho$ of $\aA_{\cA}$ starting from~$\ia^{\interI}$.
                

        \item Suppose that $\cA = \bot$ or $\cA$ is a concept name with $\cB_1 \preceq \cA$ and $\cB_2 \prec \cA$, for concept names $\cB_1$ and $\cB_2$ where we have $\cB = \cB_1 \dland \cB_2$ (the proof for $\cB = \cB_2 \dland \cB_1$ is analogous).
        Then, our derivation restricted to the first $d-1$ elements is also a $(\tboxT, \ell)$-derivation of both $\cB_1(\ia)$ and $\cB_2(\ia)$ from $\aboxA$ of length $d-1$.
        Hence, we can invoke the inductive hypothesis to obtain an accepting run $\rho_1 \deff (\ia^{\interI}, \delta_0, \dd_1), \ldots, (\dd_{k-1}, \delta_{k-1}, \dd_{k})$ of $\aA_{\cB_1}$ starting from~$\ia^{\interI}$ and an accepting run $\rho_2$ of $\aA_{\cB_2}$ starting from~$\ia^{\interI}$.
        We construct the desired run of $\aA_{\cA}$ starting from $\ia^{\interI}$ as follows. We first apply an \ref{trans:aut}-transition from $q_0 \deff (\{\top\},\cA)$ to $q_1 \deff (\{\top, \cB_2\},\cA)$, followed by a \ref{trans:noc}-transition to reach $q_2 \deff (\{\top, \cB_2\},\cB_1)$, then apply \ref{trans:weak} to arrive at $q_3 \deff (\{\top\},\cB_1)$, from which we follow run $\rho_1$.
        Formally, the desired run of $\aA_{\cA}$ starting from~$\ia^{\interI}$ is 
          $(\ia^{\interI}, q_0, \aA_{\cB_2}?, q_1, \ia^{\interI}), (\ia^{\interI}, q_1, \top?, q_2, \ia^{\interI}), (\ia^{\interI}, q_2, \top?, q_3, \ia^{\interI}), \rho_1$.

        
    \item $\cA$ is a concept name (or $\cA = \bot$) and $\cB = \exists{\rs}.\cC$ for $\cC$ being either $\top$, $\cA$, or a concept name of height smaller than $\ell$. Hence, there exists an individual name $\ib$ such that $\rs(\ia,\ib) \in \aboxA_{d-1}$ and $\cC(\ib) \in \aboxA_{d-1}$.
    We distinguish two cases depending on whether $\ib$ is an individual name from $\aboxA$ or not.

    \begin{enumerate}
        \item[4a.] Suppose $\ib \in \indA$ (\ie $\ib$ is not anonymous); the other case will be discussed next.
      By the definition of a derivation, we also have $\rs(\ia,\ib) \in \aboxA$.
      Observe that our derivation restricted to the first $d-1$ elements is a $(\tboxT, \ell)$-derivation of $\cC(\ib)$ from $\aboxA$ of length $d{-}1$.
      Thus, by the inductive hypothesis, there exists an accepting run $\rho$ of $\aA_{\cC}$ starting from~$\ib^{\interI}$.
      To construct the desired run of $\aA_{\cA}$, we apply a \ref{trans:succ}-transition from $q_0 \deff (\{\top\},\cA)$ to $q_1 \deff (\{\top\},\cC)$, stepping via $\rs$ from $\ia^{\interI}$ to $\ib^{\interI}$, from where we follow run $\rho$.
      Note that all transitions from $\rho$ are also transitions of $\aA_{\cA}$, since $\cC$ is either $\top$, $\cA$, or a concept name of height lower than that of $\cA$, and thus the same transitions are present in both automata.
      Finally, the desired run of $\aA_{\cA}$ starting from~$\ia^{\interI}$ is simply
      $(\ia^{\interI}, q_0, \rs, q_1, \ib^{\interI}), \rho$.

            \item[4b.]
      Suppose $\ib \not\in \indA$ (\ie $\ib$ is anonymous).
By the definition of a derivation, there exists a GCI $\cD_1 \sqsubseteq \exists{\rs}.\cD_2$ that, when ``applied to'' $\ia^\interI$, triggered the creation of $\ib^\interI$ (more formally: there is an index $i$ such that $\aboxA_{i+1} = \aboxA_i \cup \{ \rs(\ia,\ib), \cD_2(\ib) \}$ and $\cD_1(\ia) \in \aboxA_i$).
We again distinguish two cases.
\begin{itemize}
  \item There is no index $j < d$ such that $\ia = \ia_j$ in the active query $\cA_j(\ia_j)$.
  Applying the first part of Lemma~\ref{lemma:main-prop-of-stratified-derivation} to our derivation, we infer $(\tboxT\rstr{\ell}, \aboxA_{d-1}\rstr{\ell-1}^{\ia}) \models \cA_{d-1}(\ia)$.
  To construct an accepting run of $\aA_\cA$ starting from $\ia^{\interI}$, it suffices to accumulate the concept names from $\aboxA_{d-1}\rstr{\ell-1}^{\ia}$ into the premise (via \ref{trans:aut} rules, which is possible due ot the height of these concepts), followed by an \ref{trans:anon}-transition that simulates the classical entailment.
  Concretely, let $\cE_1, \cE_2, \ldots, \cE_m$ be an enumeration of the concepts in $\aboxA_{d-1}\rstr{\ell-1}^{\ia}$, and define $q_i \deff (\{ \top \} \cup \{ \cE_1, \ldots, \cE_i\},  \cA )$ for $0 \leq i \leq m$. The desired accepting run of $\aA_\cA$ starting from $\ia^{\interI}$ is then $\rho \deff \rho_1, \ldots, \rho_{m+2}$, where: $\rho_i \deff (\ia^{\interI}, q_{i-1}, \aA_{\cE_i}?, q_i, \ia^{\interI})$ for all $1 \leq i \leq m$, $\rho_{m+1} \deff (\ia^{\interI}, q_m, \top?, q_{m+1}, \ia^{\interI})$,
  where $q_{m+1}$ is obtained from $q_m$ by replacing its goal component with $ \top $, which creates a final state.
\item Otherwise, let $j < d$ be some index such that $\cA_j(\ia)$ is an active query.
Observe that the derivation restricted to its first $j$ elements is a $(\tboxT, \ell)$-derivation of $\cA_j(\ia)$ from $\aboxA$ of length less than $d$, so the inductive hypothesis yields an accepting run $\rho$ of $\aA_{\cA_j}$ starting from~$\ia^{\interI}$. 
Moreover, we can assume without loss of generality that $j$ is appropriate, such that applying the second part of Lemma~\ref{lemma:main-prop-of-stratified-derivation} to our derivation gives $(\tboxT\rstr{\ell}, \aboxA_{d-1}\rstr{\ell-1}^{\ia} \cup \{\cA_j(\ia)\}) \models \cA(\ia)$.
To construct an accepting run of $\aA_\cA$ starting from $\ia^{\interI}$, it therefore suffices to accumulate the concept names from $\aboxA_{d-1}\rstr{\ell-1}^{\ia}$ into the premise via \ref{trans:aut}-rules (which is possible by the height of these concepts), then apply a \ref{trans:anon}-transition simulating the classical entailment and replacing the goal $\cA$ with $\cA_j$, then discharge all premises via a \ref{trans:weak}-transition, and finally follow $\rho$.
Formally, let $\cE_1, \cE_2, \ldots, \cE_m$ enumerate the concepts in $\aboxA_{d-1}\rstr{\ell-1}^{\ia}$, and set $q_i \deff (\{ \top \} \cup \{ \cE_1, \ldots, \cE_i\},\, \cA)$ for $0 \leq i \leq m$.
The desired accepting run of $\aA_\cA$ starting from $\ia^{\interI}$ is then $\rho' \deff \rho_1', \ldots, \rho_{m+3}', \rho$, where $\rho_i' \deff (\ia^{\interI}, q_{i-1}, \aA_{\cE_i}?, q_i, \ia^{\interI})$ for $1 \leq i \leq m$, and $\rho_{m+1}' \deff (\ia^{\interI}, q_m, \top?, q_{m+1}, \ia^{\interI})$ and $\rho_{m+2}' \deff (\ia^{\interI}, q_{m+1}, \top?, q_{m+2}, \ia^{\interI})$, where $q_{m+1}$ is obtained from $q_m$ by replacing its goal component with $ \cA_j$, and $q_{m+2}$ is obtained from $q_{m+1}$ by resetting its promise component to $\{ \top \}$, which gives us $\rho_{m+3}' \deff (\ia^{\interI}, q_{m+2}, \top?, (\{\top\},\{\cA_j\}), \ia^{\interI})$.
    \end{itemize}
    \end{enumerate}
    \end{enumerate}
    As we exhausted all the cases, we are done with the inductive step. This concludes the proof.
\end{proof}

%% file: sections/appendix-pspace.tex
\section{Appendix on Combined Complexity}

In our reduction, we employ the following lemma.
\begin{lemma}[{\cite[{Theorem 6.1}]{Schaefer78}}]
    \textsc{Quantified Boolean Formula} (\textsc{QBF}), i.e. given a propositional logic formula with quantifiers, is the formula valid, is \textsf{PSPACE}-hard, even if the formula is in prenex-3-DNF.
\end{lemma}

\begin{lemma}
    The following problem is \textsf{PSPACE}-hard (even if $\preceq$ is a linear ordering on the concept names):
    \begin{align*}
        \text{Given: } & \text{ An $\mathcal{ELI}_\preceq$-TBox $\tboxT$, an ABox $\aboxA$ with a single individual name $\ia$ and a concept name $\cD$}\\
        \text{Question: } & \text{ Does $(\aboxA,\tboxT) \models \cD(\ia)$ hold?}
    \end{align*}
\end{lemma}
\begin{proof}
    Consider the formulas $\varphi$, $\psi$, the ABox $\aboxA$, the TBox $\tboxT$, the query $\cC^{\text{True}}_0(\ia)$, the set $\Clang$ of concept names, the set $\Rlang$ of role names, and the ordering $\preceq$, all as defined in the main body of the paper.  
    We claim that $(\aboxA,\tboxT) \models \cC_{0}^{\text{True}}(\ia)$ if and only if $\varphi \equiv \textsc{True}$. Since making a copy per level of the roles $\rr_0$ and $\rr_1$ is done only to obey the ordering of the concept-/role-names, we will treat all copies of $\rr_0$ and $\rr_1$, respectively, the same.

    We make the following observations: The universal model of $(\aboxA,\tboxT)$ has the form of a binary tree of depth $n$ with $\ia$ as its root, since (1) enforces an $\rr_0$ and an $\rr_1$ child per node per level. Thus any node on any level $i \in \{0,\ldots,n\}$ can be identified by one string from $\overline{\sigma}\in\{0,1\}^j$ by matching the path $\ia \xrightarrow[]{\rr_{\sigma_1}}\circ\xrightarrow[]{\rr_{\sigma_2}}\cdots \xrightarrow[]{\rr_{\sigma_j}} \circ$ to the string $\sigma_1\sigma_2 \ldots \sigma_i \in \{0,1\}^i$ (the root $\ia$ is identified by the empty string~$\varepsilon$). Furthermore, among the ``assignment concepts'' $\{\cX_i^0, \cX_i^1 \mid i \in \{ 0, 1, \ldots, n \}\}$, the node $\sigma_1\sigma_2 \ldots \sigma_i$ satisfies exactly $\{\cX_1^{\sigma_1}, \cX_2^{\sigma_2},\ldots, X_i^{\sigma_i}\}$.
    \newcommand{\I}{\mathfrak I}
    \begin{cl}
        Every node $\sigma_1\ldots\sigma_i \in \{0,1\}^{\leq n}$ in the universal model of $(\aboxA,\tboxT)$ satisfies $\cC_i^{\mathrm{True}}$ iff. $\I_{\overline{\sigma}} \coloneqq\{x_j \mapsto \sigma_j \mid j \in [i]\} \models Q_{i+1}x_{i+1}\ldots Q_nx_n\psi(x_1,\ldots,x_i)$.
    \end{cl}
    \begin{proof}
        We show the claim via induction on $n,\ldots,0$: At level $n$, $\overline{\sigma} \coloneqq\sigma_1\ldots\sigma_n$ satisfies $\cC_n^{\mathrm{True}}$ by definition iff. there is a monomial $\lambda_{j1} \land \lambda_{j2} \land \lambda_{j3}$ s.t. $\overline{\sigma}$ satisfies $\cA_j$ and $\Lambda_{j1}$\footnote{ Say that a node satisfies $\Lambda_{jk}$ if and only if $\lambda_{jk} = x_i$ (respectively $\lambda_{jk} = \neg x_i$) and the node satisfies $\cX^1_{i}$ (respectively $\cX^0_i$).}, which is the case iff. $\overline{\sigma}$ satisfies $\Lambda_{j3}$, $\Lambda_{j2}$, and $\cA_j^1$, which is the case iff $\overline{\sigma}$ satisfies $\Lambda_{j3}$, $\Lambda_{j2}$, $\Lambda_{j2}$, and $\cL_n$ (the latter one can be taken for granted). So this is the case iff $\I_{\overline{\sigma}} \models \psi$.
        
        Assume now the claim already holds for level $i+1$, then a node $\overline{\sigma} \coloneqq\sigma_1\ldots\sigma_i$ satisfies $\cC_{i}^{\text{True},\sigma_{i+1}}$ for $\sigma_{i+1} \in \{0,1\}$  iff $\overline{\sigma}\sigma_{i+1}$ satisfies $\cC^\mathrm{True}_{i+1}$ iff $\I_{\overline{\sigma}\sigma_{i+1}} \models Q_{i+2}x_{i+2}\ldots Q_nx_n\psi(x_1,\ldots,x_{i+1})$. If $Q_{i+1} = \exists$, then $\overline{\sigma}$ satisfies $\cC_i^{\mathrm{True}}$ if it satisfies $\cC_{i}^{\text{True},0}$ or $\cC_{i}^{\text{True},1}$ and in the case of $Q_{i+1} = \forall$ only if it satisfies both, in either case completing the inductive step.
    \end{proof}
    This claim applied to the root now demonstrates our Lemma.
\end{proof}